\documentclass[journal,10pt]{IEEEtran}

\usepackage{amsmath}
\usepackage{amssymb}
\usepackage{algorithm}
\usepackage{cite}
\usepackage{color}
\usepackage{soul}

\usepackage[dvipsnames]{xcolor}

\newtheorem{thm}{Theorem}
\newtheorem{lem}{Lemma}
\newtheorem{prop}{Proposition}
\newtheorem{cor}{Corollary}
\newtheorem{proof}{Proof}

\ifCLASSINFOpdf
\usepackage[pdftex]{graphicx}
\usepackage{graphicx, epstopdf}
\else
\usepackage[dvips]{graphicx}
\graphicspath{{../eps/}}
\DeclareGraphicsExtensions{.eps}
\fi

\begin{document}

	\title{{Stochastic Geometry Modeling of Cellular V2X Communication Over Shared Channels}}

	
	\author{\IEEEauthorblockN{Muhammad Nadeem Sial\IEEEauthorrefmark{1},
			Yansha Deng\IEEEauthorrefmark{1},~\IEEEmembership{Member,~IEEE},
			Junaid Ahmed\IEEEauthorrefmark{2},\\
			Arumugam Nallanathan\IEEEauthorrefmark{3},~\IEEEmembership{Fellow,~IEEE} and Mischa Dohler \IEEEauthorrefmark{1},~\IEEEmembership{Fellow,~IEEE}}\\
		\IEEEauthorblockA{\IEEEauthorrefmark{1}Department of Informatics, King's College London, London, UK}\\
		\IEEEauthorblockA{\IEEEauthorrefmark{2} COMSATS University, Islamabad, Pakistan}\\
		\IEEEauthorblockA{\IEEEauthorrefmark{3}Queen Mary University of London , London, UK}
		
		\thanks{ 
			Corresponding author: Yansha Deng (email: yansha.deng@kcl.ac.uk.)}}
	

	%
	
	{}

	\maketitle

	\begin{abstract}
		To overcome the limitations of Dedicated Short Range Communications (DSRC) with short range, non-supportability of high density networks, unreliable broadcast services, signal congestion and connectivity disruptions, cellular vehicle-to-everything (C-V2X) communication networks, standardized in 3rd Generation Partnership Project (3GPP) Release 14, have been recently introduced to cover broader vehicular communication scenarios including vehicle-to-vehicle (V2V), vehicle-to-pedestrian (V2P) and vehicle-to-infrastructure/network (V2I/N). In C-V2X, vehicles can directly communicate over {PC5 based dedicated sidelinks called direct mode or V2V communication.} However, high vehicle densities may require reuse of cellular spectrum for V2V. Moreover, infrastructure mode communication through V2I/N links can augment V2V communication by enhancing communication range and reliability for enhanced safety along with consistent performance under traffic congestions. Motivated by the stringent connection reliability, {spectral efficiency, and coverage requirements in C-V2X}, this paper presents the first comprehensive and tractable analytical framework for performance of C-V2X networks over shared V2V and cellular uplink channels, {where the transmitting vehicles can deliver their information via infrastructure or direct mode, based on their distances, propagation environments and the bias factor. By practically modeling the vehicles on the roads using the doubly stochastic Cox process and the base-stations, we derive new association probabilities, new success probabilities of infrastructure and direct mode, and overall success probability of the C-V2X communication over shared channels, which are validated by the simulations results. Our results reveal the benefits of our proposed model (possibility of selecting both direct and infrastructure modes over shared channels) compared to V2V network  in terms of success probability.}
		
		
	\end{abstract}

	
	

	\IEEEpeerreviewmaketitle

	\section{Introduction}
	
	\label{Intro}
	\IEEEPARstart{T}HE automobile industry is evolving toward connected and autonomous vehicles that offer many benefits, such as improved road safety, less traffic congestion, reduce environmental impacts, lower capital expenditure and additional traveler information services \cite{V2XTo5G}. A key enabler of this evolution is vehicle-to-everything (V2X) communication, which allows a vehicle to communicate with other vehicles (V2V), nearby infrastructure (V2I), cellular-based networks (V2N) and even pedestrians (V2P) \cite{V2XTo5G}. The V2X communications based on cellular infrastructure, referred as Cellular V2X (C-V2X) have been defined by the Third Generation Partnership Project (3GPP) group \cite{3GPPRelease14}. This innovation promises to eliminate 80\% of the current road accidents and help in fostering auto-mobile and telecommunication industries for a smarter and safer ground transportation system \cite{harding2014vehicle}. 
	
	
	Existing V2V communication can be supported via the DSRC standard, however, it has certain limitations such as short range (about 300 meters), unable to support high density of networks and has unreliable broadcast services \cite{abboud2016interworking}. The system relies on road side units (RSUs), which are not currently deployed at all locations especially in rural areas due to longer distances. The underlying carrier sensing multiple access (CSMA) medium access control (MAC) protocol also exhibits signal congestion and connectivity disruptions due to rapid changing network topology and ad-hoc vehicular networks \cite{ren2015power}. More importantly, single DSRC technology cannot support a variety of incoming vehicular oriented applications.
	
	
	{To augment DSRC communication, C-V2X communication based on 4G LTE was proposed to meet vehicular communication capacity, latency and coverage requirements} \cite{3GPPRelease14}, \cite{abboud2016interworking}, \cite{zheng2015heterogeneous} and \cite{sun2014d2d}. {In future releases, 3GPP is working on specifying 5G-based V2X} \cite{V2XTo5G}. {This C-V2X has several key advantages over DSRC, including longer range and enhanced reliability, resulting in enhanced safety, more consistent performance under traffic congestions, evolution path towards 5G for emerging applications and better coexistence with other technologies. Moreover, LTE and 5G can be used for RSU functions thus eliminating the need for highway authorities to install and maintain RSUs which also helps to reduce costs and accelerate the realization of C-V2X.}

	
	The C-V2X technologies seek to address a variety of safety use cases such as forward collision warning, emergency electronic brake light (EEBL), control loss warning, blind spot and lane change warning, as well as vulnerable road user (VRU) safety applications \cite{V2XCellSol}. It can also support queue warning, hazardous road condition warning, automated parking and tolling systems, traffic advisories and dynamic ride sharing, infotainment, local information, route planning, map dissemination and fleet management. {The transmission of C-V2X messages for these applications can be done via a direct mode or infrastructure mode (V2I/N)} \cite{V2XCellSol}. {In the direct mode, two or more vehicles can communicate directly with each other. In the infrastructure mode, vehicles communicate through cellular network or RSUs referred as V2I/N links. In this mode, the messages sent to the network or RSU in uplink (UL) may be intended for a C-V2X application server or may be intended to other vehicles nearby, in which case the network or RSU can forward the packet to the vehicles in downlink over a larger area. Therefore, a given vehicle may use direct and infrastructure modes interchangeably for same application as shown in Fig.} \ref{Fig1ProposedModel} {(links sharing the spectrum are shown with additional box)}. {It is also envisaged that due to large number of vehicles on roads, there can be scarcity of cellular spectrum allocated for PC5 based dedicated sidelink.} Therefore, vehicular communication can occur on shared cellular channels (uplink or downlink frequencies) to improve spectral efficiency. To model and analyze these C-V2X communication over shared channels, stochastic geometry has been proposed, considering that it has been utilized as a powerful tool to model and analyze mutual interference between transceivers in the wireless networks, such as conventional cellular networks \cite{novlan2013analytical,andrews2011tractable}, wireless sensor networks \cite{deng2016physical}, cognitive radio networks \cite{deng2016artificial,elsawy2013stochastic}, and heterogeneous cellular networks \cite{deng2016modeling,HElSawy2014Stochastic,sial2017analysis,sial2017novel,sial2019realistic}.



	\begin{figure}[!t]
		\centering
		\includegraphics[scale=0.28]{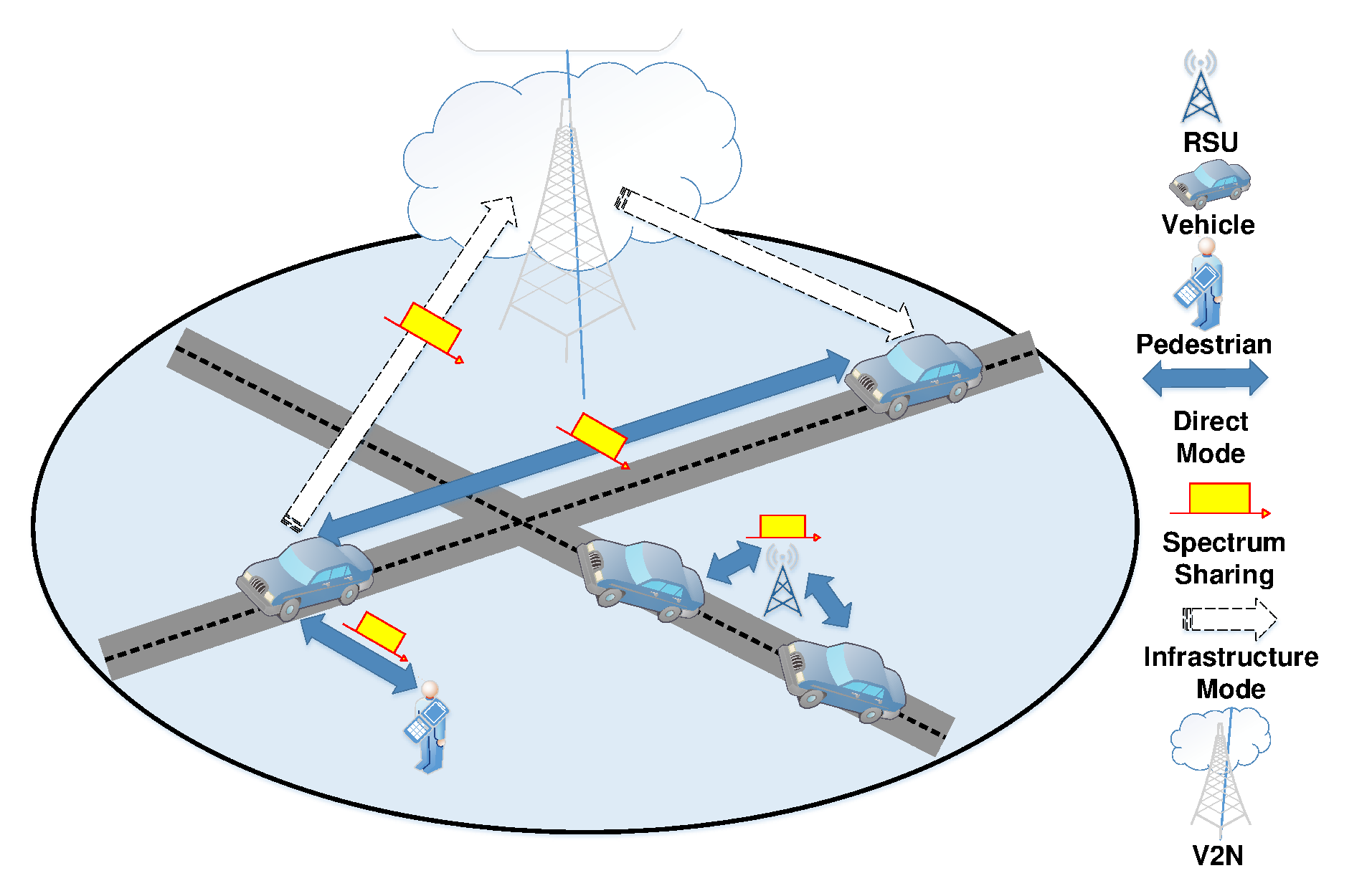}
		\caption{{Cellular V2X communication having direct and infrastucture modes.}}
		\label{Fig1ProposedModel}
	\end{figure}

	
	The initial studies on vehicular communication have focused on 
	modeling the V2V communication (i.e. without involvement of the cellular nodes) using stochastic geometry \cite{blaszczyszyn2009maximizing, busanelli2011performance,farooq2016stochastic, ni2015packet, steinmetz2015stochastic,jeyarajreliability}, where simple spatial models 
	with a single road, a multi-lane road, or  orthogonal roads were considered. The works in \cite{voss2009distributional,gloaguen2010analysis,morlot2012population,gloaguen2005simulation, blaszczyszyn2015random,chetlur2017coverage} accounted for the randomness of roads distributions. In \cite{voss2009distributional}, the nodes in the WiFi mesh networks were modeled by a Cox process on a Poisson-Line tessellation (PLT), and the nodes on each line are modeled by a inhomogeneous 1D PPP, where the probability density function of the shorted Euclidean distance between two inter-nodes was derived. Later on in \cite{gloaguen2010analysis}, the Cox process on a PLT was generalized to  a Poisson-Line tessellation (PLT), Poisson-Voronoi tessellation (PVT), or a Poisson-Delaunay tessellation (PDT), and the nodes on each line are modeled by a homogeneous 1D PPP.  Their results have shown that PLT often gains preference over PVT and PDT in modeling road systems\footnote{ It has also been used in other related applications, such as in modeling the effect of blockages in localization networks \cite{aditya2017asymptotic}.} due to its analytical tractability. In \cite{morlot2012population}, the uplink coverage probability was derived for a network where the typical receiver is randomly chosen from a PPP, and the locations of transmitter mobile users alongside roads are modeled as a Cox process on a Poisson line process (PLP). 
	In \cite{chetlur2017coverage,chetlur2018success}, the coverage probability of the V2V communication was derived, where the transmitters and receivers were modeled using independent Cox processes on the same PLP (i.e. a doubly-stochastic spatial model), and it captures the irregularity in the spatial layout of roads via the PLP model, and the distribution of vehicles on each road via  the 1D PPP model.   
	
	%
	

	Note that \cite{voss2009distributional,gloaguen2010analysis,morlot2012population,gloaguen2005simulation, blaszczyszyn2015random,chetlur2017coverage} are limited to the V2V communication or V2I communication. The first performance characterization of the  C-V2X downlink communication was studied in the master thesis in \cite{guha2016cellular}, where  the association probabilities and the coverage probabilities for the  V2V and the base station to vehicle downlink communications were derived for  the maximum power based association scheme and the threshold distance based association scheme. Recently, \cite{choi2018analytical} has performed downlink coverage analysis of cellular network leveraging vehicles where authors have derived the distance of typical receiver vehicle at center from nearest base station or vehicle. Further, they derived downlink association and the coverage probabilities of the typical vehicle in terms of integral formulas to characterize the downlink performance of vehicular communication. {However, in these works, V2V communication is sharing frequencies with cellular downlink and accordingly, interference characterization has been performed. Moreover, messages transmission from a vehicle to cellular network on the uplink has not been considered for infrastructure mode of C-V2X.}     
	
	%
	%

	{The present paper can be seen as an extension to downlink communication work presented by} \cite{guha2016cellular,choi2018analytical,chetlur2018success}. {In this paper, we focus on the C-V2X communication driven by the stringent high reliability, coverage and spectral efficiency requirements for safety, traffic management and infotainment applications. In this paper, we model the C-V2X communication where messages can be exchanged directly between nearby vehicles in direct mode or vehicles can send messages to nearby vehicles through V2I/N links in infrastructure mode using V2X application server to achieve larger coverage area as envisaged in C-V2X. Therefore, in this paper, we propose sharing of V2V and cellular uplink frequency bands for C-V2X communication in direct and infrastructure modes to improve spectral efficiency. {In our proposed model, the direct mode communication can share the frequencies with cellular uplink instead of downlink as defined by 3GPP Release 14} \cite{3GPPRelease14}. Hence, interference characterization presented in this paper is different from downlink analysis proposed by} \cite{guha2016cellular,choi2018analytical,chetlur2018success}. {Due to the shared spectrum between the direct mode links and cellular uplinks, the receiving vehicle or base-station (BS) communication with the nearest transmitting test vehicle will be interfered by the communication of all other vehicles transmitting on a particular frequency resource. In this case, no interference from cellular base-stations will be observed by the receiving vehicle or BS as downlink is not being shared by direct mode or cellular uplink.} Furthermore, authors in \cite{guha2016cellular,choi2018analytical,chetlur2018success} {have not considered messages transmission from a vehicle to cellular network on uplink channels to route traffic from a given vehicle to nearby vehicles through cellular network which is required in infrastructure mode of C-V2X. To achieve this, we calculate association probabilities for selecting the direct and infrastructure modes which has not been previously presented in the existing literature.} Additionally, we suggest a mechanism during association process to control vehicular communication over cellular network to limit interference levels to cellular users. In our system model, the locations of vehicles are modeled as a Cox process on a Poisson Line Process (PLP), and the cellular BSs of V2N are deployed as 2D PPP. Our contributions can be summarized as follows:
	\begin{itemize}
		\item {We present a comprehensive and tractable analytical framework for analyzing the cellular V2X communication, where the vehicles decide to transmit to other vehicles in direct mode or cellular BS in infrastructure mode via shared V2V and cellular uplink carriers depending on their corresponding distances, propagation environments and association bias.} 
		\item Based on the proposed C-V2X mode selection scheme, we derive the shortest distances and the association probabilities of the vehicles using direct or infrastructure mode via the shared V2V and cellular uplink carriers, respectively. 
		\item {We derive the analytical expressions for the  success probabilities of the V2V communication, V2B and the B2V communication, and the overall success probability of receiving vehicle (i.e., C-V2X communication) using shared V2V and cellular uplink resources, which are validated by Monte Carlo simulation.} 		
		\item In C-V2X network with low and medium vehicle intensity on the roads, there is almost equal probability of transmission via the direct or infrastructure modes. Moreover, infrastructure mode success probability increases at faster rate with the increase of vehicle nodes, and the success probability of the C-V2X communication in direct mode improves with increasing the road intensity.
		\item Our results have shown that the success probability of the C-V2X communication {while using both direct and infrastructure modes is comparable with  that of the V2V communication alone and network mode facilitates vehicular communication over longer ranges without the requirement of RSUs.} 		
		
	\end{itemize}
	
	The rest of the paper is organized as follows. The mathematical preliminary, system model along with assumptions, and the methodology of analysis are described in Section II and III, respectively. Section IV presents the analysis of association probabilities of C-V2X communication. The success probability is analyzed in Section V. Section VI presents and discusses the numerical and simulation results. The paper is concluded in Section VII. A list of the key mathematical notations used in this paper is given in Table I. 
	
	

	\begin{figure}[!t]
		\centering
		\includegraphics[scale=0.195]{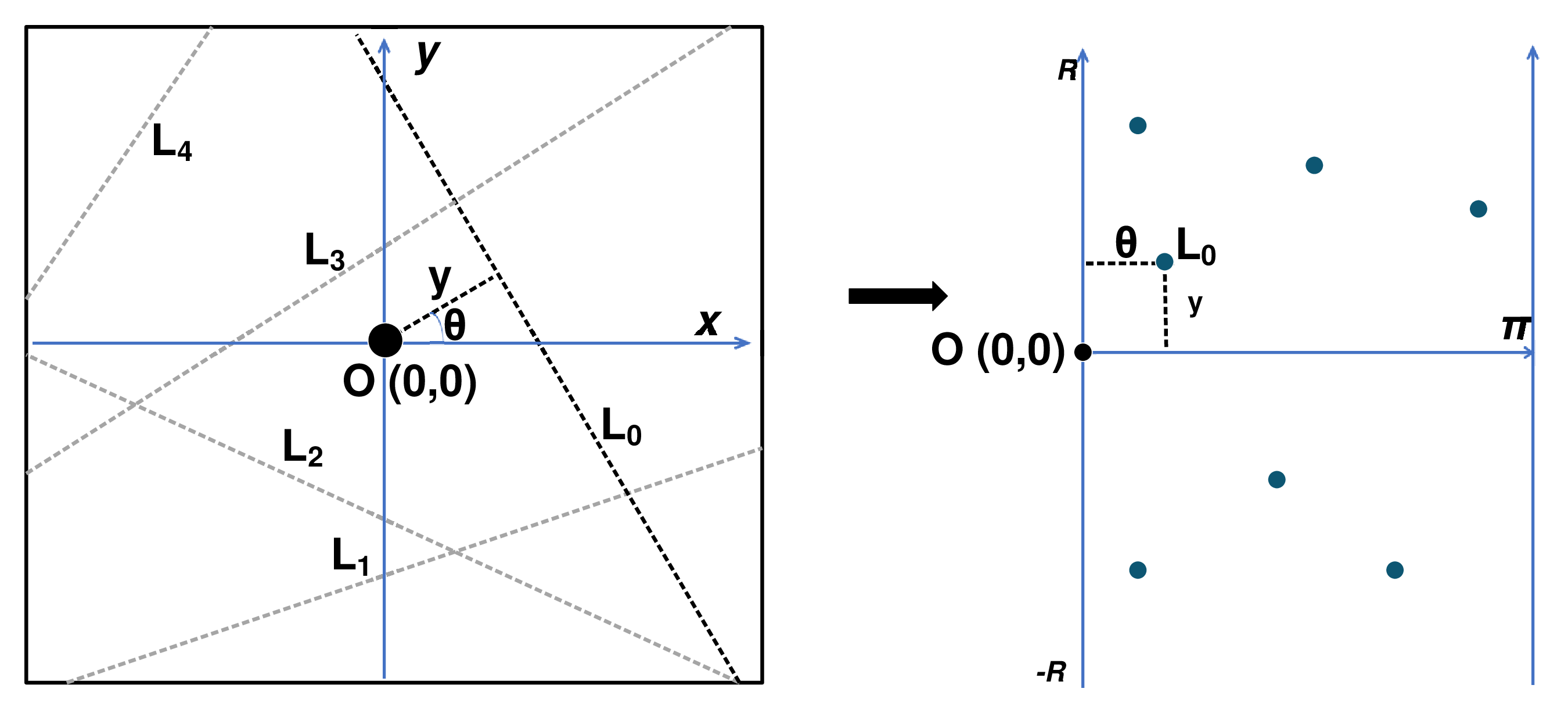}
		\caption{(a) Illustration of Poisson Line Process in two dimensional plane ${\mathbb{R}^2}$ (left). (b) Illustration of {poisson process} on representation space  { $\mathbb{C}=[0,2\pi) \times [0,\infty)$} (right).}
		\label{PLPIllustration}
	\end{figure}
	
	\section{ Preliminary: Poisson Line Process}
	\label{PLPDef}
	
	The V2X networks exhibit unique spatial characteristics due to the fact that vehicles are only driven on roadways, which are predominantly linear in nature {and layout of the roads is often irregular, which makes it possible to model the road system as a realization of a line process \cite{baccelli1997stochastic,gloaguen2010analysis,voss2009distributional,morlot2012population,chetlur2017coverage}. Therefore,} we model the roadways as a network of lines that are distributed on the plane according to a Poisson Line Process (PLP). In this section, we provide a brief introduction of PLP, the detailed information of the underlying theory can be found in \cite{chiu2013stochastic,guha2016cellular,chetlur2017coverage}.
	
	A Poisson line process is  a random collection of lines in a 2D plane. Any undirected line $\boldsymbol{L}$ in ${\mathbb{R}^2}$ can be uniquely characterized by its perpendicular distance $y$ from the origin $O(0, 0)$ and the angle $\theta$ subtended by the perpendicular dropped onto the line from the origin with respect to the positive x-axis in counter clockwise direction, as shown in Fig. \ref{PLPIllustration}. The pair of parameters $\theta$ and $y$ can be represented as the coordinates of a point on the cylindrical surface { $\mathbb{C}=[0,2\pi) \times [0,\infty)$ as illustrated in Fig. \ref{PLPIllustration}.} Clearly, there is a one-to-one correspondence between the lines in ${\mathbb{R}^2}$ and points on the cylindrical surface $ \mathbb{C}$. Thus, a random collection of lines  can be constructed from a set of points on $\mathbb{C
	}$. In other words, the set of points generated by a PPP with certain density on $\mathbb{C}$ correspond to the PLP with the same density for lines on  ${\mathbb{R}^2}$. 
	
	For a PLP $\phi_{R}$ with the intensity $\lambda_R$ {within a circular region $\mathcal{B}(0,R)$, where radius $R$ $\in$ ${\mathbb{R}}$}, the corresponding points are independent and uniformly distributed in representation space $\mathbb{C}= [-R,R] \times [0,\pi]$  with  a surface area of $2\pi R$. Thus, the expected number of points in the PPP that lie in $\mathbb{C}$ is $2\pi \lambda_R R$, and the number of lines {intersecting} a disc of radius $R$ is a Poisson distributed with mean $2\pi \lambda_R R$. In PLP, the values of $\theta$ and $y$ of each line follow a uniform distribution over an appropriate range defined by $\mathbb{C}$. In this work, we limit ourselves to motion-invariant PLP  for analytical simplicity \cite{chetlur2017coverage}, where the line process is invariant to the rotation of axes to the origin. The PLP is also considered to be stationary, where translated line process $T\phi_{R} =$ $\left\{ {T(L_1), T(L_2),....} \right\}$ of PLP, $\phi_{R} = \left\{ {L_1, L_2,....} \right\}$ has the same distribution of lines as that of $\phi_{R}$ for any translation $T$ in the plane \cite{chetlur2017coverage}.

	\begin{figure}[!t]
		\centering
		\includegraphics[scale=0.28]{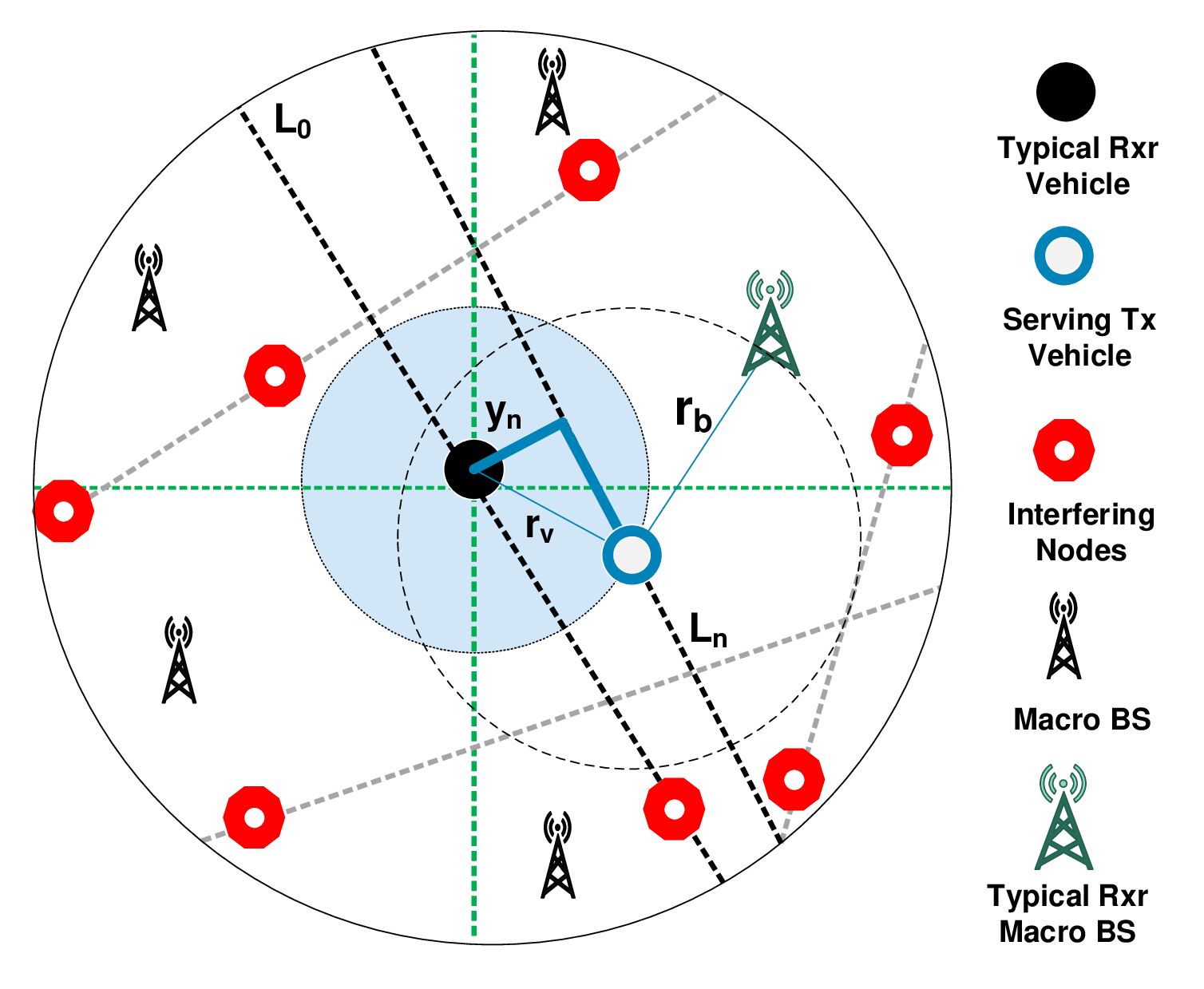}
		\caption{Illustration of the system model.}
		\label{Fig2SysModel}
	\end{figure}
	


	\section{System Model}
	\label{sec:Sysmodel}
	In this work, we consider a cellular V2X network with coexistence of V2V, V2P and V2I/N communications as per scenarios defined by 3GPP Release 14 and shown in Fig. \ref{Fig1ProposedModel}. In our system model, vehicle to vehicle messages can be exchanged through direct or infrastructure modes of C-V2X. The system model is described in detail in the following subsections.
	
	{\subsection{C-V2X Vehicular Nodes, Pedestrian and RSUs}
		As mentioned in Section \ref{PLPDef}, we model the roads as motion-invariant PLP $\phi_R$ with line intensity, $\lambda_R$ as per details given in \cite{choi2018analytical, guha2016cellular, chetlur2017coverage}, thus  the intensity of equivalent Poisson Point Process (PPP) on the representation space $\mathbb{C}$ is $\lambda_R$. The \emph{C-V2X vehicular nodes} are randomly distributed on each road as homogeneous 1D PPP with intensity $\mu_n$. {As RSUs are installed near roadside, the infrastructure mode communication through RSUs can be seen as two direct mode communication instances due to nature of deployment near roads (vehicle to RSU and RSU to vehicle). Moreover, pedestrian are also required to communicate with vehicles once they are near roads. Therefore, in this paper, we distribute pedestrian and RSUs on each road according to independent PPP's of intensities $\mu_p$ and $\mu_r$, receptively. The resulting distribution is still a PPP with intensity $\mu_v = \mu_n + \mu_r + \mu_p$ as discussed by} \cite{chetlur2017coverage,guha2016cellular}. However, the pedestrian and RSUs near roadside will have channel conditions different from vehicular nodes due to varying antenna heights which have been ignored in this paper for analytical tractability. For simplicity, the V2V, V2P and V2I communications are collectively called the direct mode communication for C-V2X throughout rest of the paper. {We use the term V2V communication throughout the paper to represent vehicle to vehicle communication without mode selection.}   
		
		%
		%
		%

		
		Assuming that each vehicle transmits independently with a probability $p$ {on a particular frequency resource}, the locations of transmitting vehicles on each road is then given by a thinned PPP with intensity, $\mu_t = p\mu_v$, and we denote the set of locations of the transmitting vehicles {operating at a particular frequency channel} on a line $L$ by {${W_L}$}. Correspondingly, the distribution of receiving vehicles on each line is also a thinned PPP with intensity, $\mu_r = (1-p)\mu_v$.
		In other words, the  transmitting and receiving vehicles are modeled as the doubly stochastic processes called Cox processes, $\phi_t$ and $\phi_r$, which are driven by the same PLP, $\phi_R$. 
		
		\begin{table}[htbp]\label{table1}
			\protect\caption{Notations}

			\begin{tabular}{|c|l|l|}
				\hline 
				Notations & Definition \tabularnewline
				\hline 
				\hline 
				$\lambda_R$ &Intensity of roads, 2D PLP  \tabularnewline
				\hline 
				$\mu_v$ &Intensity of C-V2X nodes, ID PPP \tabularnewline
				\hline 
				$\lambda_b$ &Intensity of base-stations, 2D PPP\tabularnewline
				\hline
				$\phi_R$ & Poisson Line Process (PLP) for roads\tabularnewline
				\hline
				$\phi_t$ & Cox process for transmitting nodes\tabularnewline
				\hline
				$\phi_r$ & Cox process for receiving nodes\tabularnewline
				\hline
				$\phi_b$ & 2D PPP for cellular base-stations\tabularnewline
				\hline
				$B$ & Association bias, $0 \text{ to } \infty$  \tabularnewline
				\hline 
				$P_{\mathbb{\mathit{b}}}$  &Base-station transmit power\tabularnewline
				\hline 
				$P_{\mathbb{\mathit{v}}}$ & Vehicle transmit power\tabularnewline
				\hline 
				$\alpha_v$ & Path loss exponent for V2V link or direct mode \tabularnewline
				\hline 
				$\alpha_b$& Path loss exponent for network mode \tabularnewline
				\hline 
				$y_n$ & Perpendicular distance of road from origin\tabularnewline
				\hline 
				$\theta$ & Angle of road from x-axis\tabularnewline
				\hline 
				$BW$ & C-V2X communication channel bandwidth, $10 \text{ MHz}$\tabularnewline
				\hline 
				V2V & Vehicle to vehicle link alone \tabularnewline
				\hline 
				V2B & Vehicle to base-station uplink in network mode \tabularnewline
				\hline 
				B2V & Base-station to vehicle downlink in network mode \tabularnewline
				\hline 
				$r_v$ & V2V distance \tabularnewline
				\hline 
				$r_b$ & V2B uplink distance  \tabularnewline
				\hline 
				$r_{b2v}$ & B2V downlink distance \tabularnewline
				\hline
				$\sigma^2$ & Thermal noise \tabularnewline
				\hline   
			\end{tabular}
		\end{table}
		
		For analytical simplicity, we can translate the origin $O=(0,0)$ to the location of the typical receiver vehicle. The translated point process $\phi_{r_{0}}$ can be treated as the superposition of the point process $\phi_r$, an independent 1D PPP with intensity $\mu_r$ on a line passing through the origin and a vehicle at the origin $O$. {This can be realized according to the following steps defined by \cite{chetlur2017coverage}.  We first add a point at the origin to the PPP in the representation space $\mathbb{C}$ by applying the Slivnyak's theorem \cite{chiu2013stochastic}}, thereby obtaining a PLP $\phi_{R_{0}}= \phi_{R} \cup {L_0}$ with a line $L_0$ passing through the origin in $\mathbb{R}^{2}$ and second, we add a point at the origin to the 1D-PPP on the line ${L_0}$ passing through the origin in $\mathbb{R}^{2}$ {by applying the Slivnyak's theorem \cite{chiu2013stochastic}}. The line, $L_0$ passing through the origin is referred as typical line in this paper. Since, both $\phi_t$ and $\phi_r$ are driven by the same line process, the translated point process $\phi_{t_{0}}$ is also the superposition of $\phi_t$ and an independent PPP with intensity $\mu_t$ on ${L_0}$. Note that the other receiving vehicles in the network do not interfere with the typical receiving vehicle in our setup, therefore, we focus only on the distribution of transmitting vehicles {which are operating on a particular frequency channel}. {In this work, the impact of the vehicle mobility and direction are neglected, similar to the work presented by \cite{tong2016stochastic,choi2018analytical}.}
		
		
		\subsection{V2N Network for Infrastructure Mode}
		
		The V2N segment of infrastructure mode consists of cellular macro base-stations (BSs). In this paper, the cellular BSs are spatially distributed in $\mathbb{R}^{2}$ according to the 2D PPP with intensity $\lambda_b$. This model for macro BSs has been validated to be as accurate to the typical hexagonal grid model \cite{jo2012heterogeneous}. We assume that each base-station in V2N will always have at least one vehicular node connected with it in the uplink. {Due to the shared spectrum between the direct and V2N links, interference will exist between them. For simplicity, the V2N communication in infrastructure mode are referred as network mode throughout rest of the paper.}

		
		\subsection{C-V2X Mode Selection Scheme}
		In our model, the vehicles transmit with the fixed power $P_v$. The mode selection between direct and network modes is determined by the corresponding V2V or V2B link distances, propagation environments and association bias factor $B$ to limit interference and traffic on cellular network. In other words, it can tune the trade-off between interference and data offloading \cite{elsawy2014analytical}. In this flexible C-V2X mode selection scheme, the vehicle selects the direct mode if $B{r_v^{-\alpha{_v}}} \geq r_{b}^{-\alpha{_b}} $, otherwise the vehicle selects the V2B link or network mode, where $r_v$ is the shortest  V2V link distance, and $r_b$ is the distance between the transmitting vehicular node and its closest cellular BS of V2N network. {The C-V2X mode selection is primarily being decided based on V2V link and cellular uplink distances whereas cellular downlink for receiving vehicle is being selected independent of uplink.} {Note that, it is possible to use independent associations of DL and UL transmissions, but there may exist location dependence due to the location of base stations in certain situations. Hence, the associations of DL and UL may be correlated even though their association rules are different. In this paper, the DL and UL related correlations due to location of base-stations, have been ignored and approximated results have been obtained in accordance with existing literature} \cite{lee2015hybrid,akbar2016simultaneous}. Note that for the extreme case with $B = 0$, the direct mode will never be selected, whereas for the case with $B =\infty$, each vehicle transmitter always selects the direct mode. It is worth mentioning that one main advantage of this mode selection criterion is that it brings an inherent interference protection to the cellular uplink, which is necessarily required to enable vehicular communication over shared channels. 
		
		
		\subsection{{Cellular Downlink Selection for Network Mode}}
		
		{In our model, message received by V2N cellular BS from transmitting vehicle in the uplink is delivered to nearby receiving vehicle on downlink. The downlink is established between receiving vehicle and cellular BS of V2N network based on maximum received power. In such case, message transmissions in uplink and reception in the downlink can be through two different cellular BSs. However, the downlink frequencies are different from the frequencies being used by direct mode communication or by uplink of V2N. Therefore, there will be no interference from vehicular nodes on downlink. However, cellular base-stations excluding the base-station with which receiving vehicle is associated will interfere in the downlink as per model presented by} \cite[Section-III]{andrews2016primer }.

		\subsection{Channel Model}
		A general power-law path-loss model is considered in which the signal power decays at the rate, $r^{-\alpha}$ with the propagation distance $r$, where $\alpha$ is the path-loss exponent. Due to the different propagation environments experienced in the network and direct modes, each type of mode is given its own path-loss exponent, namely, $\alpha_b$ and $\alpha_{v}$, respectively.
		The small-scale channel fading is modeled as slow-flat Rayleigh fading {as used by \cite{guha2016cellular,choi2018analytical,tong2016stochastic,farooq2016stochastic,elsawy2014analytical,jeyarajreliability}}, where its channel gain is assumed to be exponentially distributed with unit mean.  All the channel gains are assumed to be independent of each other, independent of the spatial locations, symmetric, and are identically distributed (i.i.d.). {We use Rayleigh fading for network and direct modes for its analytical tractability and this assumption is also the most popular in the literature to get closed-form expressions \cite{elsawy2013stochastic}. However, analytical expressions can be easily extended to interesting scenarios such as non-exponential fading\cite{choi2018analytical}.} For log normal shadowing, we have included the shadowing in a transparent way by using the displacement theorem given by \cite{baccelli2009stochastic,dhillon2014downlink}.  

		
		{\subsection{C-V2X Communication Reliability or Success Probability}}
		The reliability of C-V2X communication depends upon coverage probabilities of V2V link, cellular uplink and downlink. The overall success probability of C-V2X can be calculated by using the total probability law depending upon individual coverage probabilities of communication links and association probabilities of selected mode. In our model, the transmitting vehicle connects to a receiver (either cellular BS or vehicle) depending on the corresponding distances, propagation environments and association bias, and thus the  transmitting vehicle can operate in modes {$M \in \left\{ {v2b,v2v} \right\}$}, where $v2b$ and $v2v$ modes denote the shared link communication between vehicle and vehicle through cellular BSs in network mode and that between vehicle and vehicle in direct mode, respectively. In our model, once the vehicle operates in $v2b$ mode or selects cellular uplink, then cellular downlink or $b2v$ link is also selected whose success probability is calculated independent of uplink. The success probability (reliability) of the arbitrary receiver (base-station or vehicle) for V2V link or cellular uplink conditioned on minimum distance $r$ between transmitter and receiver and mode $M$ can be defined as the probability that the SINR of receiver is greater than a SINR threshold, $z$, which is given as
		\begin{equation}\label{SINR}
		P_M^{S}(z|r,M) = {\rm Pr}\left[ {\frac{{{P}  {r^{ - \alpha }}  h}}{{I + {\sigma ^2}}} > z} \right],
		\end{equation}
		where $P$ is transmit power of vehicle, $h$ is channel gain and $r$ is the minimum distance between transmitter and the receiver. {Similarly, the success probability for cellular downlink can be defined as the probability that the SINR of receiver vehicle is greater than a SINR threshold, $z$, which is given as} 
		
		\begin{equation}
		P_{b2v}^{S}(z|r_{b2v}) = {\rm Pr}\left[ {\frac{{{P_b}  {r_{b2v}^{ - \alpha }}  h}}{{I + {\sigma ^2}}} > z} \right],
		\end{equation}
		
		{where $P_b$ is transmit power of BS, $h$ is channel gain, $I$ is interference from BSs other than the selected BS and $r_b2v$ is the minimum distance between transmitter BS and the receiver vehicle.} 
		
		{In our system model, we translate the origin to the location of typical receiver vehicle. Then, we find the distance  $r_v$ between nearest transmitting test vehicle and typical receiver vehicle, which is shortest V2V link distance. Later, we find the distance ($r_b$) between the selected transmitting test vehicle and nearest BS which is the shortest V2B link distance. The transmitting test vehicle selects the direct mode if $B{r_v^{-\alpha{_v}}} \geq r_{b}^{-\alpha{_b}}$ otherwise the transmitting test vehicle selects the V2B link or network mode, where $r_v$ is the shortest V2V link distance, and $r_b$ is the shortest distance between the transmitting vehicular node and its closest cellular BS of V2N network. {In such case, the selected link distance ($r$) between transmitter (vehicle) and receiver (BS or vehicle) is minimum. Generally, there could be interfering transmitter inside the disk due to V2V or transmitter-receiver pairs other than selected link (link between receiver at origin and nearest transmitter). However, in our system model, receiving vehicles other than receiving vehicle at center have been ignored, therefore, there will be no interfering vehicles present in the disk $\mathcal{B}(0, r)$ with radius $r$} as shown in Fig.} \ref{Fig2SysModel}. {For this setup, we derive the coverage probability of receiver (cellular BS or vehicle), which connects to its closest transmitting test vehicle in the network at distance $r$, having no interfering vehicles present in the disk $\mathcal{B}(0, r)$.} The size of $\mathcal{B}(0, r)$ affects the distribution of interfering nodes, ${\Phi _I}$ as well as the interference to receiver. We need to characterize the distribution of interferers, ${\Phi _I}$ as per their locations (interferers located on road passing through origin or on all other roads) to determine the Laplace transform of the distribution of interference
		power conditioned on the serving distance $r$ and mode $M$. 
		
		Remind that $r$ is the minimum possible distance between the receiver and transmitter, and the interfering vehicles are located outside the disk $\mathcal{B}(0, r)$ {as per considered system model}. Therefore, the interferers can be broadly divided into two categories, (a) the interfering vehicles that are located outside the disk $\mathcal{B}(0, r)$ on the road passing through the origin and (b) the interfering vehicles located outside the disk $\mathcal{B}(0, r)$ on all other roads. The distribution of interferes located on all other roads can be denoted as ${\Phi _v|\mathcal{B}(0, r)}$ and the disk $\mathcal{B}(0, r)$ does not contain any interfering vehicles. Similarly, the distribution of interferes located on road passing through the origin can be denoted as ${\Phi _r|r}$ and there are no vehicles on the road segment from $0$ to $2r$ for road passing through origin. As such, the success probability given in Eq. (\ref{SINR}) can be rewritten as
		\begin{equation}\label{PsGen}
		P_M^{S}(z|r,M) = {\rm Pr}\left[ {\frac{{P  {r^{ - \alpha }}  h}}{{{I_v} + {I_r} + {\sigma ^2}}} > z} \right],
		\end{equation}
		where ${{\sigma ^2}}$ is the thermal noise power, ${I_v}$ and ${I_r}$ are the aggregate interference due to the vehicles located on all other roads, and on the road passing through the origin that operate in mode $v2b$ or $v2v$, respectively.  
		
		\section{{Association Probabilities for Mode Selection}}
		To facilitate the reliability analysis of proposed C-V2X communication over shared V2V and UL channels, we first derive the distance distributions of the vehicle to vehicle link in direct mode and the V2B link in network mode, and their corresponding association probability in the following.
		
		
		%
		%
		
		\begin{lem}[V2V Distance]
			\label{Lemma1}
			The CDF of the shortest distance $R_v$ between the vehicular transmitter and an arbitrary vehicular receiver in the V2V link, ${F_{{R_v}}}({{r_v}})$ is given in \cite{guha2016cellular}, we still present results below for completeness
			\begin{align}
			{F_{{R_v}}}({r_v}) &= 1 - \left[ {\exp \left( { - 2\pi {\lambda _R}\int\limits_{{y_n} = 0}^{{r_v}} {1 - {e^{ - 2{\mu _v}\sqrt {{r_v}^2 - {y_n}^2} }}} d{y_n}} \right)} \right.\nonumber
			\\&\Bigg. { \times \exp \left( { - 2{\mu _v}{r_v}} \right)} \Bigg].\label{cdfv2v}
			\end{align}		
		\end{lem}
		
		


		\begin{cor}[The PDF of V2V Link]
			\label{cor7}			
			The Probability Density Function (PDF), ${f_{{R_v}}}({{r_v}})$ of the shortest distance  between the vehicular transmitter and a typical vehicular receiver, $R_v$ in the V2V link is derived as 		
			\begin{align}
			{f_{{R_v}}}({r_v}) &= \Bigg( { - 4\pi {\lambda _R}{\mu _v}\int\limits_{ 0}^{{r_v}} { {\frac{{{r_v}{e^{ - 2\mu \sqrt {{r_v}^2 - {y_n}^2} }}}}{{\sqrt {{r_v}^2 - {y_n}^2} }}} } d{y_n}} - 2{\mu _v} \Bigg)
			\nonumber
			\\& \times \Bigg[   - \exp \Big( { - 2\pi {\lambda _R}}\int\limits_{ 0}^{{r_v}} {\Big( {1 - {e^{ - 2{\mu _v}\sqrt {{R^2} - {r^2}} }}} \Big)} d{y_n} \Big. \Bigg.\nonumber
			\\& \Bigg.\left. { - 2{r_v}{\mu _v}} \right)\Bigg].\label{pdfv2vfinal2}
			\end{align}
			\begin{proof}
				The Probability Density Function (PDF) of $R_{v}$ can be found by taking derivative of CDF given in Eq. (\ref{cdfv2v}), ${f_{{R_v}}}({r_v}) = \frac{d}{{d{r_v}}}\left( {{F_{{R_v}}}({r_v})} \right)$ and final result for PDF of $R_v$ is given as Eq. (\ref{pdfv2vfinal2}). The closed form solution of Eq. (\ref{pdfv2vfinal2}) is derived as 
				\begin{align}\label{pdfv2vfinalclosedform}
				{f_{{R_v}}}({r_v}) &= \Bigg( { - 2{\pi ^2}{\lambda _R}{r_v}{\mu _v}\Bigg[{I_0}(2{r_v}\mu_v) - {L_0}(2{r_v}{\mu _v})\Bigg] - 2{\mu _v}} \Bigg)\nonumber
				\\ \nonumber 
				&\times \Bigg[ { - \exp \Bigg( { - 2\pi {\lambda _R}{r_v} + {\pi ^2}{\lambda _R}{r_v}\Big[{{L_{ - 1}}(2{r_v}{\mu _v})} \Bigg.} \Big.} \Big.
				\\ 
				&\Big. {\Big. {\Big. { - {I_1}(2{r_v}{\mu _v})} \Big] - 2{r_v}{\mu _v}} \Bigg)} \Bigg],
				\end{align}
				where $I_n(z)$ ($n=0,1$) denotes the modified Bessel functions of the first kind and $L_n(z)$ ($n=0,{-1}$) denotes the modified Struve functions. This completes the proof.	
			\end{proof} 
		\end{cor}


		
		\begin{prop}[{V2B and B2V Distances}]
			\label{Lemma22}
			The PDF  of the distance between a vehicular transmitter and the nearest BS for uplink, $R_B$ for the selected link ${f_{{R_B}}}({{r_B}})$ is given in \cite[Eq. (2)]{novlan2013analytical} as
			\begin{equation} \label{V2Bpdf}
			{f_{{R_b}}}(r_b) = 2\pi {\lambda _b}r_b{e^{ - \pi {\lambda _b}{r_b^2}}}. 
			\end{equation}
		\end{prop}
		
		
		%

		
		{The distance distribution between a cellular BS and the nearest receiving vehicle in the downlink, $R_{B2V}$ for the selected link, ${f_{{R_{B2V}}}}({{r_{B2V}}})$ will be same as given in} (\ref{V2Bpdf}). In Lemma \ref{Lemma2} and Lemma \ref{Lemma3}, we derive the association probabilities of the V2V link and the  V2B  link, respectively.
		
		%
		
		\begin{lem}[{Association probability for the Direct Mode}]
			\label{Lemma2}
			The probability of the vehicular transmitter selecting the direct mode is derived as 
			\begin{align} \label{V2VProbincomplete}
			{P_{v2v}^{A}} &= \int\limits_0^\infty  { - \exp \Bigg( { - 2\pi {\lambda _R}{r_v} + {\pi ^2}{\lambda _R}{r_v}\Big[{L_{ - 1}}(2{r_v}{\mu _v})} \Bigg.} \nonumber 
			\\&\Bigg. { - {I_1}(2{r_v}{\mu _v})\Big] - 2{r_v}{\mu _v}} \Bigg)\nonumber 
			\\\nonumber&\times \Bigg( { - 2{\pi ^2}{\lambda _R}{r_v}{\mu _v}\Bigg[{I_0}(2{r_v}\mu ) - {L_0}(2{r_v}{\mu _v})\Bigg] - 2{\mu _v}} \Bigg)
			\\&\times \exp \left( { - \pi {\lambda _B}{{\left( {\frac{{r_v^{\frac{{{\alpha _v}}}
										{{{\alpha _b}}}}}}
							{{{B^{\frac{1}
											{{{\alpha _b}}}}}}}} \right)}^2}} \right)d{r_v},
			\end{align}
			where $\lambda_B$ is the cellular base-station intensity, $\mu _v$ is the C-V2X nodes intensity, $\lambda_R$ is the road intensities, $B$ is the association bias, $I_n(z)$ $(n=0, 1)$ is the modified Bessel functions of the first kind, and $L_n(z)$ $(n=0,{-1})$ is the modified Struve functions.	
			\begin{proof}
				See Appendix \ref{AppendixB}.
			\end{proof}
		\end{lem}
		
		\begin{lem}[Association probability of the V2B  Link {in Network Mode}]
			\label{Lemma3}
			The probability of the vehicular transmitter selecting V2B link is derived as
			\begin{align} \label{V2BAssProb}
			{P_{v2b}^{A}} &= \int\limits_0^\infty  {\exp \Bigg[ { - 2\pi {\lambda _R}} \Bigg.}  \times \int\limits_{{y_n} = 0}^{{B^{\frac{1}
						{{{\alpha _v}}}}}r_b^{\frac{{{\alpha _b}}}
					{{{\alpha _v}}}}} {1 - \exp \Bigg( { - 2{\mu _v}} \Bigg.} \nonumber 
			\\&\left. {\left. { \times \sqrt {{{\left( {{B^{\frac{1}
											{{{\alpha _v}}}}}r_b^{\frac{{{\alpha _b}}}
										{{{\alpha _v}}}}} \right)}^2} - {y_n}^2} } \right) \times d{y_n}} \right]\nonumber 
			\\&\times \exp \left( { - 2{\mu _v}{B^{\frac{1}
						{{{\alpha _v}}}}}r_b^{\frac{{{\alpha _b}}}
					{{{\alpha _v}}}}} \right) \times 2\pi {\lambda _b}{r_b}\exp \left[ { - \pi {\lambda _b}{r_b}^2} \right]d{r_b}	
			\end{align}
			where $\lambda_B$ is the cellular base-station intensity, $\mu _v$ is the vehicular nodes intensity, $\lambda_R$ is the road intensities, $B$ is the association bias.
			\begin{proof}
				See Appendix \ref{AppendixC}.
			\end{proof}
		\end{lem}
		
		\section{C-V2X Success Probability Over Shared {V2V and Cellular Uplink Channels}}
		In this section, we derive the success probability of the C-V2X communication where vehicles can select direct or network modes. To do so, we first need to characterize the interference from each type of interferer category ($I_v, I_r$). Thus, we calculate the general form of Laplace Transform, and we derive the expressions of Laplace Transform of interference from $I_v$ and $I_r$ in this section.  
		\subsection{Laplace Transform of Interference Under Rayleigh Fading}
		As we know that Laplace Transform of interference, $I_X$ is ${{\cal L}_{{I_X}}}\left( s \right) = {E_{{I_x}}}\left[ {{e^{ - s{I_x}}}} \right]$.
		The interfering set of vehicles ${\rm X}$ for each category of interfering vehicles $(I_v,I_r)$ based on their location can be represented as
		\begin{equation}\label{eq9}
		{I_X} = \sum\limits_{x \in X} {{P_v}} {h_x}D_{^x}^{ - \alpha},
		\end{equation}
		where $X$ represent various interference sources $(I_v,I_r)$ and $P_v$ is the transmit power of the vehicle. For a vehicle ${x \in X}$, we denote its distance to the typical receiver as ${{D_x}}$. Although the random variables ${\left[ {{D_x}} \right]_{x \in X}}$, are identically distributed, they are not independent in general \cite{novlan2013analytical}. However, authors in \cite{novlan2013analytical} have shown that this dependence is weak and we will henceforth, assume each ${{D_x}}$ to be i.i.d. Due to the different propagation environments experienced in the network and direct mode or V2V communication, each type of mode will have its own path-loss exponent and $\alpha$ can be replaced with $\alpha_b$ or $\alpha_{v}$ in (\ref{eq9}) as per selected link. The expression for ${{\cal L}_{{I_X}}}\left( s \right)$ is given as
		\begin{equation}
		{{\cal L}_{{I_X}}}\left( s \right) = {\mathbb{E}_{{I_X}}}\left[ {{e^{ - \sum\limits_{x \in X} {s{P_v}{h_x}D_{^x}^{ - \alpha }} }}} \right].
		\end{equation}
		By taking expectation over ${{h_x}, D_{x}}$, we get
		\begin{equation}
		{{\cal L}_{{I_X}}}\left( s \right) = {\mathbb{E}_{{h_x},{D_x}}}\left[ {\prod\limits_{x \in X} {\exp \left( { - s{P_v}{h_x}D_{^x}^{ - \alpha }} \right)} } \right].
		\end{equation}
		By assuming all ${{{h_x}}}$ as independent, we obtain
		\begin{equation}
		{{\cal L}_{{I_X}}}\left( s \right) = {\mathbb{E}_{{D_x}}}\left[ {\prod\limits_{x \in X} {{E_{{h_x}}}\left[ {\exp \left( { - s{P_v}{h_x}D_{^x}^{ - \alpha }} \right)} \right]} } \right].
		\end{equation}
		Based on the fact that $h \sim \exp \left( 1 \right)$, we obtain
		\begin{equation} \label{conditionalLaplace}
		{{\cal L}_{{I_X}}}\left( s \right) = {\mathbb{E}_{{D_x}}}\left[ {\prod\limits_{x \in X} {\left[ {\frac{1}{{1 + s{P_v}R_x^\alpha D_{^x}^{ - \alpha }}}} \right]} } \right].
		\end{equation}
		Now, we calculate the Laplace transform of interference for each category of road (${I_v}$ and ${I_r}$) in the following. Let us denote the outer circular region in which all roads exist as $\mathcal{B}(0,R)$, and inner circular region $\mathcal{B}(0,r)$ with radius $r$ having minimum distance between transmitter and receiver as shown in {Fig. \ref{Fig2SysModel}.} Let us denote two types of road as $R_{in}$ and $R_{out}$, where $R_{in}$ are the roads intersecting the circular region $\mathcal{B}(0,r)$, and $R_{out}$ are the roads {that} lie outside the circular region $\mathcal{B}(0,r)$ and within circular region $\mathcal{B}(0,R)$. Thus, road $R_{in}$ is located at distance, $y<r$ and in case of $R_{out}$, it is located at distance $y>r$. In this case, the interferers can be located anywhere on road $R_{out}$ as it is located outside $\mathcal{B}(0,r)$. However, in case of $R_{in}$, the interferers will be located in regions between $(-\sqrt {{R}^2 - {y}^2}, -\sqrt {{r}^2 - {y}^2})$ and 
		$(\sqrt {{R}^2 - {y}^2}, \sqrt {{r}^2 - {y}^2})$. 
		In the following, we calculate the Laplace Transform for the interferences from the vehicular transmitters located in these two types of roads (i.e., $R_{in}$ and $R_{out}$).

		\begin{cor}[Laplace Transform of Interference for Single Road Located Outside Inner Circular Region, $\mathcal{B}(0,r)$]
			\label{cor1}
			The conditional Laplace transform of interference at typical receiver, originating from a single road located at a distance $y$ ($y>r$), outside inner circular region, $\mathcal{B}(0,r)$ with radius $r$ is expressed as 
			\begin{align} \label{LaplaceSingleRoad2}
			{{\cal L}_{{I_{{R_{out}}}}}}\left( {s|r} \right) &= \exp \Bigg[ { - 2{\mu _v}}  {  \int\limits_{ 0}^\infty  {\left[ {1 - \frac{1}{{1 + s{P_v}{{\left( {{y^2} + {t^2}} \right)}^{\frac{{ - \alpha }}{2}}}}}} \right]dt} } \Bigg] .
			\end{align}
			For $\alpha=4$, the closed form solution can be simplified as 	
			\begin{equation} \label{LaplaceSingleRoad2alpha4}
			{{\cal L}_{{I_{{R_{out}}}}}}\left( {s|r} \right) = \exp \left( { - {\mu _v}\frac{{\pi {r^2}\sin \left( {\frac{1}{2}{{\tan }^{ - 1}}\left( {\frac{{{r^2}}}{{{y^2}\sqrt {\frac{1}{z}} }}} \right)} \right)}}{{\sqrt {\frac{1}{z}} \sqrt[4]{{{y^4} + {r^4}z}}}}} \right).
			\end{equation}
			\begin{proof}
				See Appendix \ref{AppendixD}.
			\end{proof}
		\end{cor}
		\begin{cor}[Laplace Transform of Interference for Single Road Intersecting the Inner Circular Region, $\mathcal{B}(0,r)$]
			\label{cor3}
			The conditional Laplace transform of interference at typical receiver, originating from a single road located at distance $y$ ($y<r$), intersecting the inner circular region, $\mathcal{B}(0,r)$ with radius $r$ is derived as 
			\begin{align} \label{LaplaceSingleRoad1}
			{{\cal L}_{{I_{R_{in}}}}}\left( {s|r} \right) &= \exp \Bigg[ { - 2{\mu _v}}\nonumber  \\&\times{  \int\limits_{ \sqrt {{r^2} - {y^2}} }^\infty  {\left( {1 - \frac{1}{{1 + s{P_v}{{\left( {{y^2} + {t^2}} \right)}^{\frac{{ - \alpha }}{2}}}}}} \right)} dt} \Bigg].
			\end{align}	
			\begin{proof}
				For this case, the value of $t$ varies from $t=(-\sqrt {{r}^2 - {y}^2})$ to $(\sqrt {{r}^2 - {y}^2})$ and the range of region in which the interfering nodes will be located is $(-\infty, -\sqrt {{r}^2 - {y}^2})$ and $(\sqrt {{r}^2 - {y}^2}, \infty)$. Therefore, the Laplace Transform of interference can be calculated by changing the limits of Eq.  (\ref{LaplaceSingleRoad2}) and is given as Eq. (\ref{LaplaceSingleRoad1}). This completes the proof. 
			\end{proof}
		\end{cor}	
		
		Based on the results in Corollary \ref{cor1} and \ref{cor3}, we can derive the Laplace transform of the interference from the vehicles located on the road passing through the origin and  that located on all other roads in Corollary \ref{cor4} and \ref{cor8}, respectively. 
		
		\begin{cor}[Laplace Transform of Interference from Vehicles Located on Road Passing Through Origin]
			\label{cor4}
			The conditional Laplace transform of interference at typical receiver, originating from road located at a distance of $y=0$ is derived as 
		\end{cor}
		\begin{equation} \label{LaplaceSingleRoadthroughcentrealphagenral}
		\begin{array}{l}
		{{\cal L}_{{I_r}}}\left( {s|r} \right) =   
		\exp \Big[ - \frac{{2 \times r \times z \times \mu_v \,{ \times _2}{F_1}\left( {1,\frac{{\alpha  - 1}}{\alpha },2 - \frac{1}{\alpha }, - z} \right)}}{{\alpha  - 1}} \Big],\alpha  > 1. \\ 
		\end{array}
		\end{equation}
		where ${_2F_1}(a,b,c,z)$ is the Hypergeometric function and $z$ is the SINR threshold. For $\alpha=4$, the closed form expression for the above equation is simplified as 	
		\begin{align} \label{LaplaceSingleRoadthroughcentrealpha4}
		{{\cal L}_{{I_r}}}\left( {s|r} \right) &= \exp \left[ {- \frac{{r \times \sqrt[4]{z}{\mu _v}}}{{\sqrt 2 }} \times \left( { - {{\tan }^{ - 1}}\left( {\frac{{\sqrt 2 }}{{\sqrt[4]{z}}} + 1} \right)} \right.} \right.\nonumber
		\\&\left. {\left. { + {{\tan }^{ - 1}}\left( {1 - \frac{{\sqrt 2 }}{{\sqrt[4]{z}}}} \right) - {{\coth }^{ - 1}}\left( {\frac{{\sqrt z  + 1}}{{\sqrt 2 \sqrt[4]{z}}}} \right) + \pi } \right)} \right].
		\end{align}
		\begin{proof}
			The conditional Laplace transform of interference can be calculated using Corollary \ref{cor3} and  $y=0$, and the resultant equation is derived as
			\begin{equation} \label{LaplaceSingleRoadthroughcentre}
			{{\cal L}_{{I_r}}}\left( {s|r} \right) = 
			\exp \left( { - 2{\mu _v}\int\limits_{ r}^\infty  {\left[ {1 - \frac{1}{{1 + s{P_v}{t^{ - \alpha }}}}} \right]dt} } \right).
			\end{equation}		
			The closed form solution of Eq.  (\ref{LaplaceSingleRoadthroughcentre}) for all values of $\alpha$ is proved in (\ref{LaplaceSingleRoadthroughcentrealphagenral}).
			
		\end{proof}

		\begin{cor}[Laplace Transform of Interference from All Roads Excluding Road Passing Through Origin]
			\label{cor8}
			The conditional Laplace transform of the total interference at the typical receiver, originating from vehicular transmitters located on all roads except the road passing through the origin is derived as 
		\end{cor}
		\begin{align} \label{LaplaceAllRoads}
		{{\cal L}_{{I_v}}}\left( {s|r} \right) &= \left[ {\exp \left( {2{\mu _v}{\lambda _R}\int\limits_{ 0}^r {1 - {{\cal L}_{{I_{R_{in}}}}}\left( s \right)dy} } \right)} \right]\nonumber
		\\&\times \left[ {\exp \left( {2{\mu _v}{\lambda _R}\int\limits_r^\infty {1 - {{\cal L}_{{I_{R_{out}}}}}\left( s \right)dy} } \right)} \right],
		\end{align}
		where ${{\cal L}_{{I_{R_{in}}}}}\left( s \right)$ and ${{\cal L}_{{I_{R_{out}}}}}\left( s \right)$ are given in \eqref{LaplaceSingleRoad1} and \eqref{LaplaceSingleRoad2}, respectively.
		\begin{proof}
			See Appendix \ref{AppendixE}.
		\end{proof}

		With the help of Corollary \ref{cor4} and \ref{cor8}, we can derive the success probabilities for the V2V link, B2V link and that for the V2B link for a given distance $r$ in the following theorems.
		\begin{thm}[Success Probability of the V2V Link {Over Shared Channels}]
			\label{Lemma6}
			The  success probability of the V2V  link for given minimum distance $r_v$ between transmitter and receiver is derived as
			
			{	 		
				\begin{align} \label{PSv2vfinal}
				{P_{v2v}^{S}}(z|r_v) &= {\exp \left( { - \frac{{z{\sigma ^2}}}{{{P_v}  r_v^{ - \alpha_v }}}} \right)  {{\cal L}_{{I_v}}}\left( {\frac{z}{{{P_v}  r_v^{ - \alpha_v }}}|r_v} \right)}\nonumber
				\\&\times {{\cal L}_{{I_r}}}\left( {\frac{z}{{{P_v}  r_v^{ - \alpha_v }}}|r_v} \right),
				\end{align}
				where ${{\cal L}_{{I_{{v}}}}}\left( s|r_v \right)$  and ${{\cal L}_{{I_{r}}}}\left( s|r_v \right)$  are given in Eqs.  (\ref{LaplaceSingleRoadthroughcentrealphagenral}) and (\ref{LaplaceAllRoads}) by substituting $s={\frac{z}{{P_v  {r_v^{ - \alpha_v }}}}}$.} 
		\end{thm}
		
		\begin{proof}
			The success probability of the V2V link over shared channels for given minimum distance, $r_v$ between transmitter and receiver and operating mode $M=v2v$ is		
			\begin{equation} \label{PSv2v1}
			{P_{v2v}^{S}}(z|r_v) = {Pr\left[ {SINR > z} \right]}.
			\end{equation}
			Using Eq.  (\ref{PsGen}), the above equation can be written as
			{	
				\begin{equation} \label{PSv2v2}
				{P_{v2v}^{S}}(z|r_v) = {Pr\left[ {\frac{{P_v  {r_v^{ - \alpha_v }}  h}}{{{I_v} + {I_r} + {\sigma ^2}}} > z} \right]}. 
				\end{equation}}
			With mathematical simplification, the final equation for the success probability of the V2V link for given minimum distance, $r_v$ between transmitter and receiver is proved in Eq. (\ref{PSv2vfinal}).
		\end{proof}

		\begin{thm}[Uplink Success Probability for the V2B Link {in Network Mode}]
			\label{Lemma7}
			The uplink or V2B success probability over shared channels for given minimum distance, $r_b$ between transmitter and receiver is derived as
			{
				\begin{align} \label{PSv2bfinal}
				{P_{v2b}^{S}}(z|r_b) &= {\exp \left( { - \frac{{z{\sigma ^2}}}{{{P_v} \times r_b^{ - \alpha_b }}}} \right) \times {{\cal L}_{{I_v}}}\left( {\frac{z}{{{P_v} \times r_b^{ - \alpha_b }}}|{r_b}} \right)}\nonumber
				\\&\times {{\cal L}_{{I_r}}}\left( {\frac{z}{{{P_v} \times r_b^{ - \alpha_b }}}|{r_b}} \right).
				\end{align}
				where ${{\cal L}_{{I_{{v}}}}}\left( s|r_b \right)$  and ${{\cal L}_{{I_{r}}}}\left( s|r_b \right)$  are given in Eqs.  (\ref{LaplaceSingleRoadthroughcentrealphagenral}) and (\ref{LaplaceAllRoads}) by substituting $s={\frac{z}{{P_v  {r_b^{ - \alpha_b }}}}}$.} 
			\begin{proof}
				The  success probability for the V2B  link for given minimum distance, $r_b$ between transmitter and receiver is presented as 
				\begin{equation} \label{PSv2b1}
				{P_{v2b}^{S}}(z|r_b) = {P\left[ {SINR > z} \right]}.
				\end{equation}
				By using Eq. (\ref{PsGen}), the above equation can be written as
				{
					\begin{equation} \label{PSv2b2}
					{P_{v2b}^{S}}(z|r_b) = {Pr\left[ {\frac{{P_v  {r_b^{ - \alpha_b }}  h}}{{{I_v} + {I_r} + {\sigma ^2}}} > z} \right]}.
					\end{equation}}
				With mathematical simplification, the final equation for the success probability of the V2B  link for given minimum distance, $r_b$ between transmitter and receiver is proved in Eq. (\ref{PSv2bfinal}).	
			\end{proof}
		\end{thm}
		
		\begin{lem}[{Downlink Success Probability for B2V Link in Network Mode}]
			\label{Lemma10}
			{The downlink success probability for B2V link is given in \cite{andrews2016primer}, we still present results below for completeness} 
		\end{lem}
		
		\begin{align}\label{PSb2v}
		P_{b2v}^S(z) &= \int\limits_0^\infty  {\exp \left( { - \frac{{z{\sigma ^2}}}{{{P_b} \times r_{b2v}^{ - {{\alpha _b}}
			}}}} \right) \times {{\cal L}_{{I_B}}}\left( {\frac{z}{{{P_b} \times r_{b2v}^{ - {{\alpha _b}}
				}}}|{r_{b2v}}} \right)} \nonumber \\ 
		&\times {f_{{R_{B2V}}}}({r_{B2V}}) \times d{r_{B2V}}, 
		\end{align}

		{where laplace transform of interference ${{\cal L}_{{I_{{B}}}}}\left( s|r_{b2v} \right)$ is given as}  
		
		\begin{align}\label{b2vlaplace}
		{L_{{I_B}}}\left( {s|r} \right) = \exp \left( { - 2\pi {\lambda _b}\int\limits_r^\infty  {\left[ {1 - \frac{1}{{1 + s{P_b}{t^{ - {\alpha _b}}}}}} \right]t \times dt} } \right),
		\end{align}
		where $s = \frac{z}{{{P_b}{r^{ - {\alpha _b}}}}}$ and $r=r_{b2v}$. The pdf ${f_{{R_{B2V}}}}({{r_{B2V}}})$ is given as
		
		\begin{equation} \label{B2Vpdf}
		{f_{{R_{b2v}}}}(r_{b2v}) = 2\pi {\lambda _b}r_{b2v}{e^{ - \pi {\lambda _b}{r_{b2v}^2}}}. 
		\end{equation}
		
		\begin{cor}[Success Probability of C-V2X {with Mode Selection}]
			\label{cor5}
			{The success probability of cellular  V2X network over shared V2V and cellular UL channels having both direct and network modes is derived as}
			{
				\begin{align} \label{PSOverallfinal}
				{P_{V2X}^{S}}(z)  &= \int\limits_0^\infty  {\exp \left( { - \frac{{z{\sigma ^2}}}{{{P_v} \times r_v^{ - \alpha_v }}}} \right) \times {{\cal L}_{{I_v}}}\left( {\frac{z}{{{P_v} \times r_v^{ - \alpha_v }}}|r_v} \right)}\nonumber
				\\&\times {{\cal L}_{{I_r}}}\left( {\frac{z}{{{P_v} \times r_v^{ - \alpha_v }}}|r_v} \right) \times {P_{v2v}^{A}}(v2v|r_v)\nonumber \\&\times {f_{{R_v}}}({r_v})d{r_v} +P_{b2v}^S(z)\nonumber
				\\&\times\int\limits_0^\infty  {\exp \left( { - \frac{{z{\sigma ^2}}}{{{P_v} \times r_b^{ - \alpha_b }}}} \right) \times {{\cal L}_{{I_v}}}\left( {\frac{z}{{{P_v} \times r_b^{ - \alpha_b }}}|{r_b}} \right)}\nonumber
				\\&\times {{\cal L}_{{I_r}}}\left( {\frac{z}{{{P_v} \times r_b^{ - \alpha_b }}}|{r_b}} \right) \times {P_{v2b}^{A}}(v2b|r_b)\nonumber\\&\times {f_{{R_b}}}({r_b}) \times d{r_b},
				\end{align}
				{where ${{\cal L}_{{I_{{v}}}}}\left( s|r \right)$  and ${{\cal L}_{{I_{r}}}}\left( s|r \right)$  are given in Eqs.}  (\ref{LaplaceSingleRoadthroughcentrealphagenral}) and (\ref{LaplaceAllRoads}) by substituting $s={\frac{z}{{P_v  {r^{ - \alpha }}}}}$ and $r=r_v$ or $r=r_b$ and $\alpha=\alpha_b$ or $\alpha=\alpha_{v}$. The ${P_{v2v}^{A}}(v2v|r_v)$ is given in Eq. (\ref{V2Vpdfconditioned}), and ${P_{v2b}^{A}}(v2b|r_b)$ is given in Eq. (\ref{V2Basseq}).} The downlink success probability $P_{b2v}^S(z)$ for B2V link in network mode is given {in Eq.} (\ref{PSb2v}).
			
			\begin{proof}		
				By using total probability law, the success probability of C-V2X network is
				\begin{align} \label{PSoverall}
				{P_{V2X}^{S}}(z) &= {P_{v2v}^{S}}(z|r_v)\times{P_{v2v}^{A}}(v2v|r_v)\nonumber
				\\&+ {P_{v2b}^{S}}(z|r_b)\times{P_{v2b}^{A}}(v2b|r_b)\times P_{b2v}^S(z),
				\end{align}
				where ${P_{v2v}^{S}}(z|r_v)$ is given in Eq. (\ref{PSv2vfinal}), ${P_{b2v}^{S}}(z)$ is given in Eq. (\ref{PSb2v}) and ${P_{v2b}^{S}}(z|r_b)$ is given in Eq. (\ref{PSv2bfinal}), ${P_{v2v}^{A}}(v2v|r_v)$ is given in Eq. (\ref{V2Vpdfconditioned}) and ${P_{v2b}^{A}}(v2b|r_b)$ is given in Eq. (\ref{V2Basseq}). By removing the distance ($r_v$ and $r_b$) condition on Eq. (\ref{PSoverall}), the C-V2X overall success probability is proved in Eq. (\ref{PSOverallfinal}), which  completes the proof.		
			\end{proof}	
		\end{cor}
		
		For the purpose of comparison, we derive the success probability of vehicular V2V communication without mode selection over shared V2V and cellular UL channel in the following:
		\begin{cor}[{Success Probability of the V2V Communication only Over Shared Channel}]
			\label{cor6}
			{The success probability of V2V Communication only, over shared V2V and UL channels is derived as}
			{		
				\begin{align} \label{PSv2vwithoutcellualr}
				{P_{v2v}^{only}}(z) &= \int\limits_0^\infty  {\exp \left( { - \frac{{z{\sigma ^2}}}{{{P_v} \times r_v^{ - \alpha_v }}}} \right) \times {{\cal L}_{{I_v}}}\left( {\frac{z}{{{P_v} \times r_v^{ - \alpha_v }}}|r_v} \right)}\nonumber
				\\&\times {{\cal L}_{{I_r}}}\left( {\frac{z}{{{P_v} \times r_v^{ - \alpha_v }}}|r_v} \right) \times {f_{{R_v}}}({r_v})\times d{r_v},
				\end{align}
				where 	the PDF of $r_v$ is given in Eq. (\ref{pdfv2vfinal2}), and ${{\cal L}_{{I_{{v}}}}}\left( s|r_v \right)$  and ${{\cal L}_{{I_{r}}}}\left( s|r_v \right)$  are given in Eqs.  (\ref{LaplaceSingleRoadthroughcentrealphagenral}) and (\ref{LaplaceAllRoads}) by substituting $s={\frac{z}{{P_v  {r_v^{ - \alpha_v }}}}}$.}
			\begin{proof}
				The success probability of V2V communication over shared channel can be derived by removing condition on $r_v$ in Eq. (\ref{PSv2vfinal}) and the final expression is proved in Eq. (\ref{PSv2vwithoutcellualr}). 	
			\end{proof}
			
		\end{cor}

		\section{Numerical Results}

		In this section, the association probabilities of mode selection scheme (direct or network mode) are plotted using (\ref{V2VProbincomplete}) and (\ref{V2BAssProb}). {The success probabilities of C-V2X over shared channels having both direct and network modes are plotted using} (\ref{PSOverallfinal}) in comparison with success probability of the V2V communication alone over shared V2V and cellular UL channels. We also plot the success probability of direct mode, and network mode along with its association probability using the first part of Eq. (\ref{PSOverallfinal}), and the second part of Eq. (\ref{PSOverallfinal}), respectively. The analytical results are validated by Monte Carlo simulations as shown in each figure. In all the figures, we set the path loss at $\alpha_{v}=\alpha_b=4$ and the thermal noise spectral density, $\sigma^2=$ -174 dBm/Hz for 10 MHz bandwidth. The transmit power of vehicles are set to be {23 dBm as defined by \cite{3GPPRelease14}}. For comparison purposes, V2V communication without mode selection scheme has also been plotted using (\ref{PSv2vwithoutcellualr}) to exhibit advantages of C-V2X communication having both direct and network modes, where modes are selected based on mode selection scheme. {In the figures, ``analyt.” represents analytical plot, ``sim" represents simulation plot, ``V2V" represents only V2V communication without mode selection scheme, ``Network mode" and ``Direct mode" represents C-V2X communication in network and direct modes, respectively. The ``C-V2X" represents cellular V2X communication having both direct and network modes.} {Note, in our model, V2V and direct mode communication shares frequencies with uplink cellular bands which is different from PC5 based dedicated sidelink proposed by 3GPP Release 14.} {It is pertinent to highlight that a small gap between simulation and analytic results of figures (Fig.} \ref{PSProb} to Fig. \ref{Bias}) {may exist due to ignored DL and UL correlation effects as a result of location of base-stations.} 
		

		
		%
		
		\subsection{Impact of the SINR threshold}
		
		In this subsection, we examine the effect of SINR threshold, $z$ on the success probability of the proposed model. In Fig. \ref{PSProb}, we set {$\lambda_R$ = 1 km/km$^2$, $\mu_v$ = 10 nodes/km, $\lambda_b$ = 20 BSs/km$^2$ and $B = 1$.} 
		Fig. \ref{PSProb} plots the  success probability of the C-V2X communication with mode selection and V2V communication without mode selection at the arbitrary receiver versus the SINR threshold. Following insights are observed: 1) the cellular V2X communication with mode selection performs equivalent to V2V communication without mode selection. {Both direct and network modes can contribute in achieving the better reliability and coverage for C-V2X network. The network mode can facilitate communication over longer ranges without the requirement of RSUs. 2) we see that contribution of network mode is significant in C-V2X network over V2V network  when the vehicle intensities are low or medium. Thus, network mode can provide longer coverage for vehicular communication which in not possible through V2V communication alone.  3) we observe that in a highly dense C-V2X network, most of the vehicles connect via direct mode instead of network mode.}
		
		%
		\begin{figure}
			\begin{center}
				\includegraphics[scale=0.42]{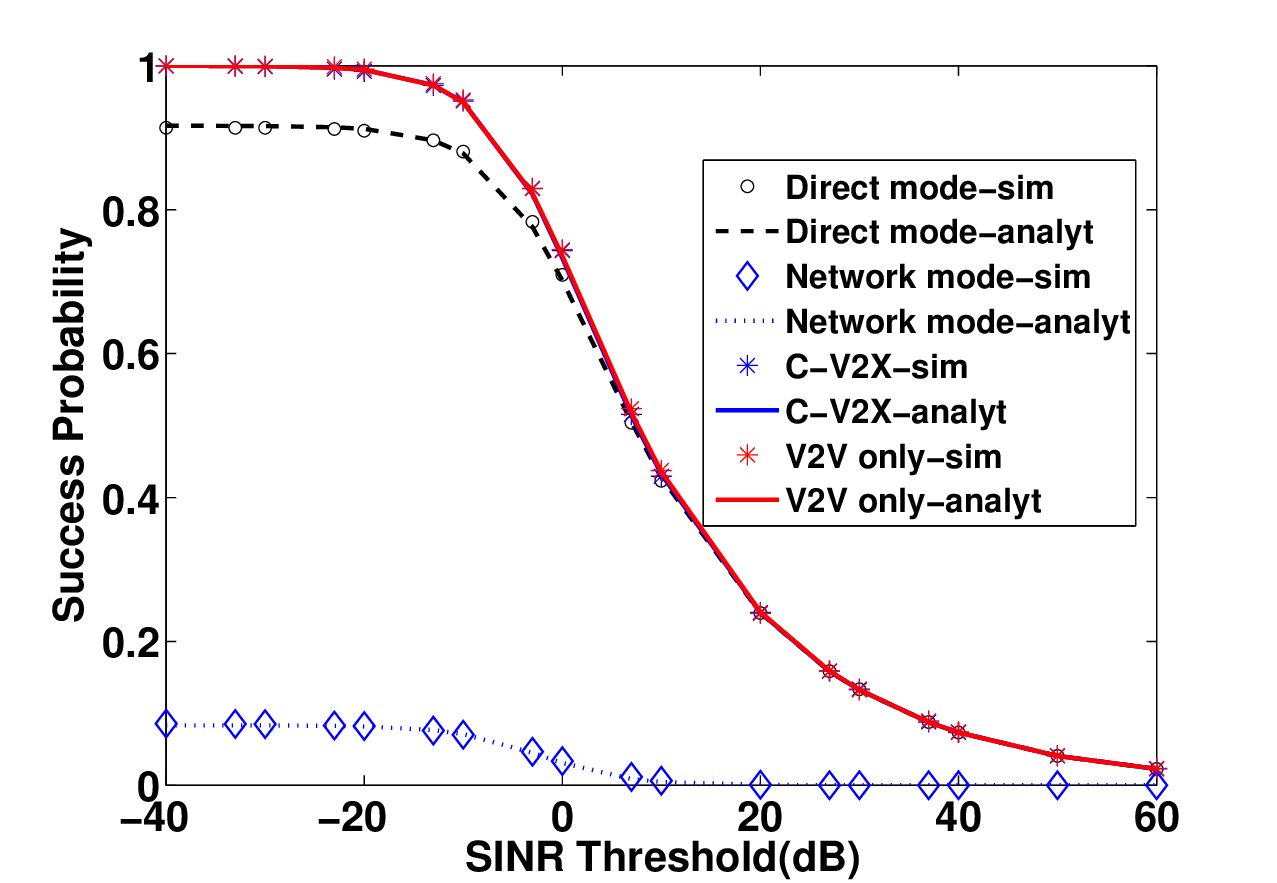}
				\caption{{The success probability versus the SINR threshold.}}
				\label{PSProb}
			\end{center}
		\end{figure}
		
		


		\subsection{Impact of the vehicle intensity}
		
		In this subsection, we examine the effect of vehicle intensity, $\mu_v$ at  success probability and association probabilities of C-V2X network. In Fig. \ref{VehVarOverall}, \ref{AssProb}, {we set $\lambda_R$ = 5 km/km$^2$, $\lambda_b$ = 20 BSs/km$^2$, $z = 0$ dB and $B = 1$.} 
		
		Fig. \ref{VehVarOverall} plots the success probability at the arbitrary receiver versus the intensity of vehicular nodes. The following insights can be observed: 1) The success probability of C-V2X communication slightly increases as the intensity of vehicles on the roads increases due to decrease in network mode success probability. 2) In low and medium intensity networks, there is almost equal probability of connecting through direct or network modes. However, in highly dense networks, this advantage of cellular V2X communication over V2V communication alone reduces to certain extent. 3) We see that the direct mode success probability increases with the increase of vehicle nodes. This is because of reduction in distance between the arbitrary receiver and vehicle located on the lines that are closer to the origin. However, the distance between arbitrary receiver and vehicles located on the lines that are farther away from the origin does not decrease at same rate due to effect of perpendicular distance of roads. This increases the desired signal power at a faster rate than the interference power, thus improving the SINR and hence the success probability at the arbitrary receiver.

		%
		\begin{figure}
			\begin{center}
				\includegraphics[scale=0.42]{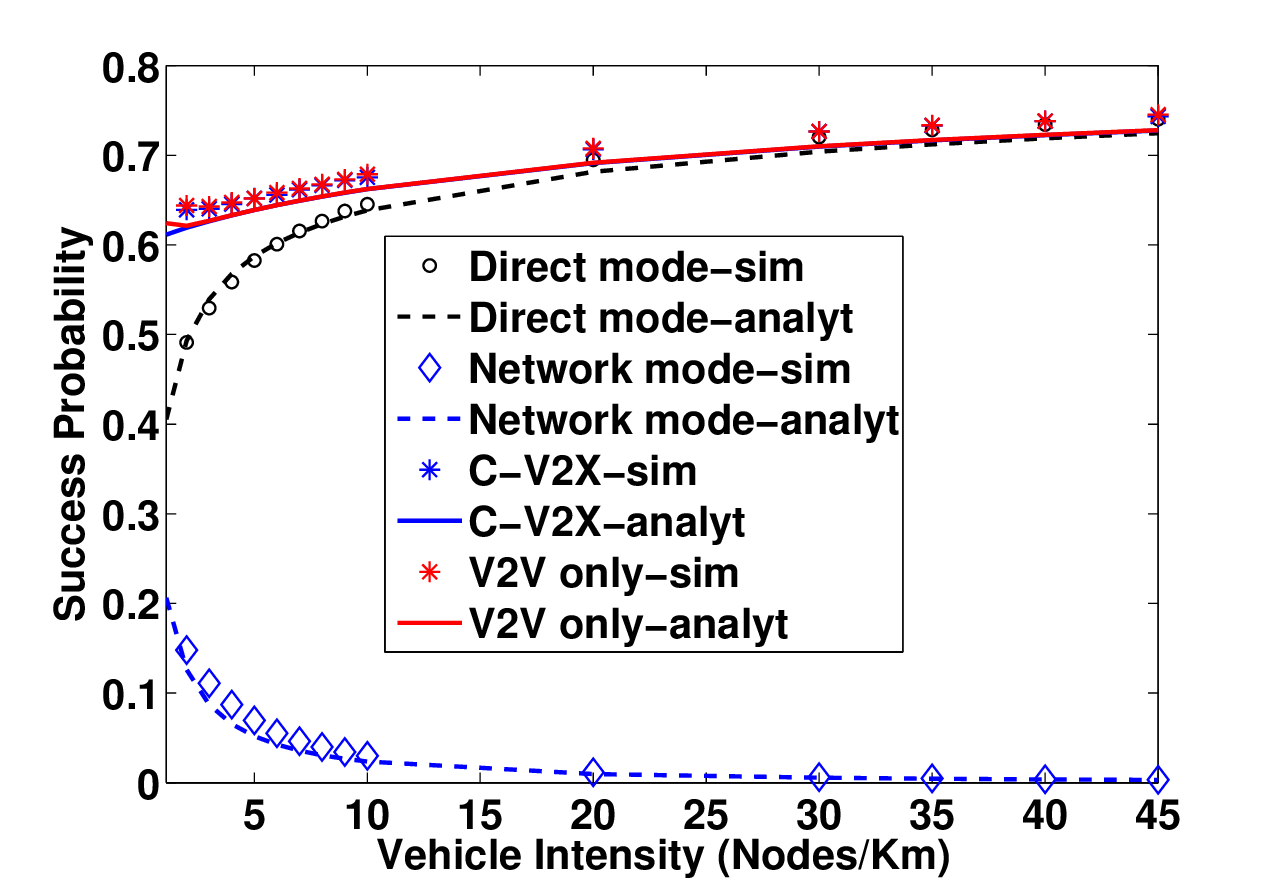}
				\caption{{The success probability versus the vehicle intensity.}}
				\label{VehVarOverall}
			\end{center}
		\end{figure}
		
		In Fig. \ref{AssProb}, we plot the association probabilities of direct and network modes versus the intensity of vehicles. From the figure, we observed that with the increase of vehicular node intensities, more vehicles start to use the direct mode as the distance between the vehicles reduces in highly dense network. {However, in a low and medium intensity C-V2X networks, both direct and network modes are used by the vehicles.}
		
		\begin{figure}
			\begin{center}
				\includegraphics[scale=0.42]{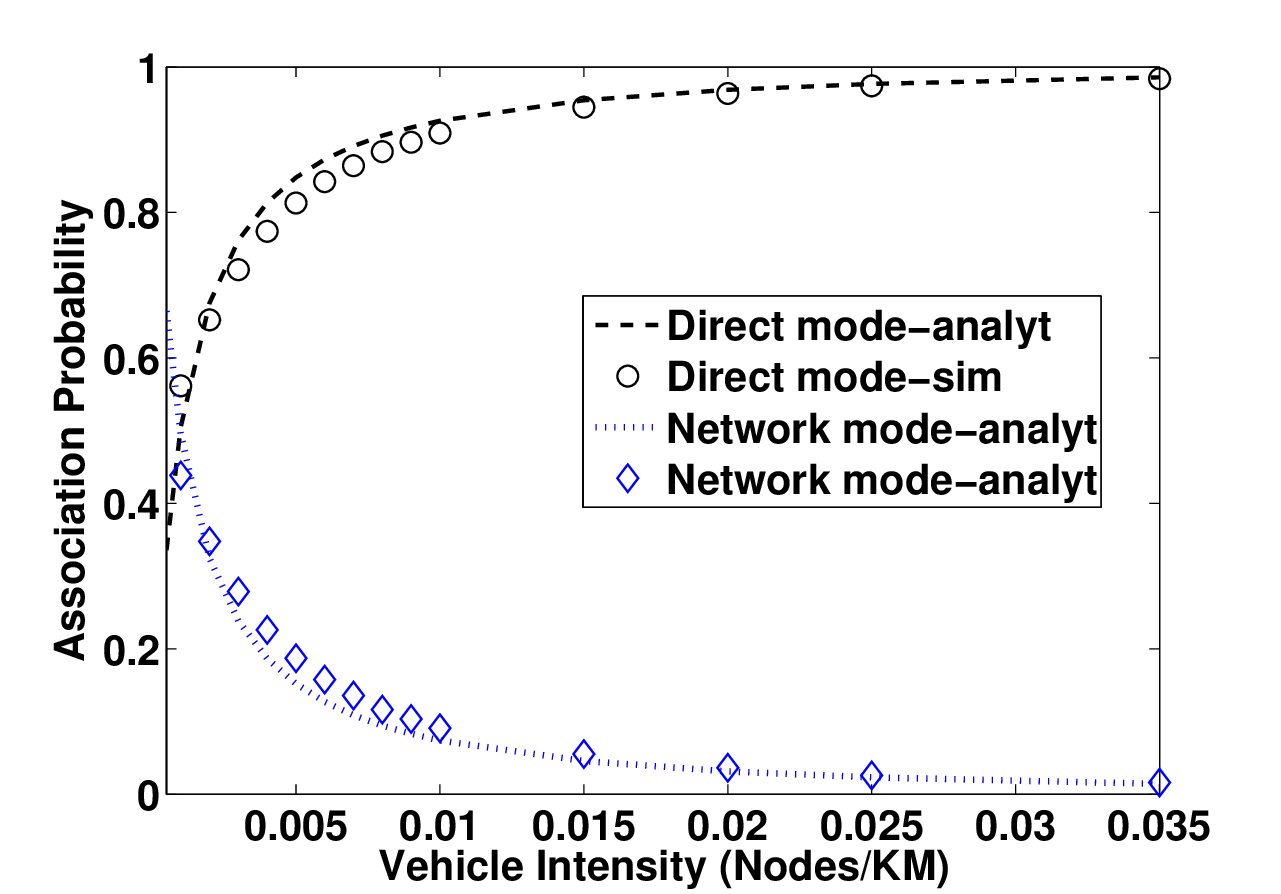}
				\caption{{The {association probability} versus the vehicle intensity.}}
				\label{AssProb}
			\end{center}
		\end{figure}
		

		\subsection{Impact of the road intensity}
		In this subsection, we examine the effect of road intensity, $\Lambda_R$ at success probability of C-V2X network having mode selection between direct and network modes and V2V communication without mode selection. In Fig. \ref{RoadVarOverall}, we set {$\mu_v$ = 5 nodes/km, $\lambda_b$ = 20 BSs/km$^2$, $B = 1$ and $z=0$ dB.} 
		
		Fig. \ref{RoadVarOverall} plots the success probability at the arbitrary receiver versus the intensity of roads. The following insights can be observed: 1) {The success probability of C-V2X communication with mode selection is equal to V2V communication only.} 2) From detailed analysis of results, we see that with the increase of road intensity, the V2V serving distance decreases by bring nodes closer due to which success probability increases and vehicles start using direct mode instead of the network mode.

		
		\begin{figure}
			\begin{center}
				\includegraphics[scale=0.42]{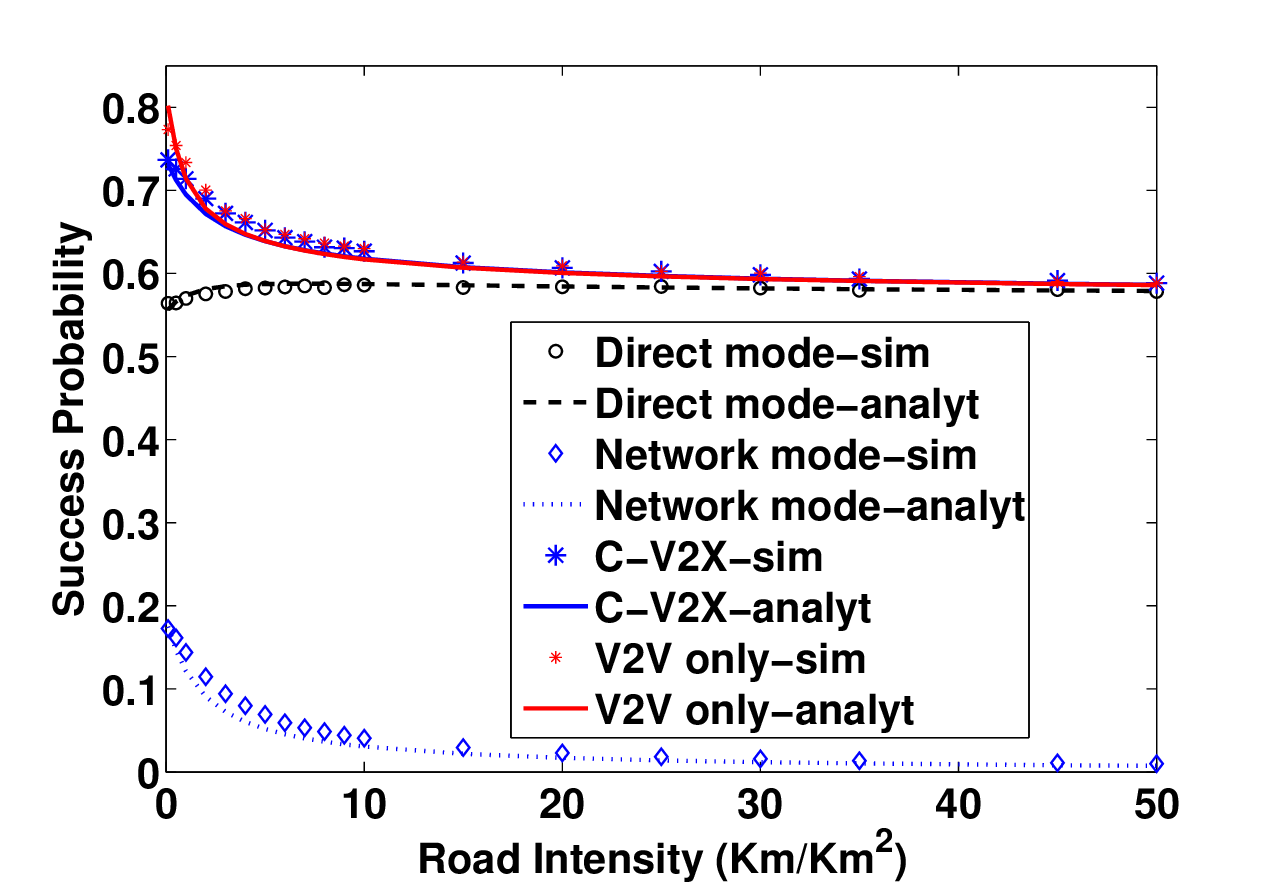}
				\caption{{The success probability versus the road intensity.}}
				\label{RoadVarOverall}
			\end{center}
		\end{figure}
		

		\subsection{Impact of the base-station intensity}
		In this subsection, we examine the effect of BS intensity, $\Lambda_b$ at the success probability of the C-V2X communication and solely V2V communication. In Fig. \ref{BSVarOverall}, we set {$\mu_v$ = 5 nodes/km, $\lambda_R$ = 5 km/km$^2$, $B = 1$ and $z=0$ dB.} 
		
		Fig. \ref{BSVarOverall} plots the  success probability at the arbitrary receiver versus the intensity of BSs. The following insights can be observed: 1) the success probability of cellular V2X network remains constant as success probability of network mode increases at same rate at which direct mode success probability decreases. 2) we observed that with the BS densification, the probability of having base-station in near vicinity to vehicular transmitter is higher than the PLP based V2V link distance, because the BSs are uniformly distributed instead of non-uniform PLP distribution of roads. Therefore, probability of connecting with network mode increases with increase of BS density. 3) the success probability of the V2V communication remains constant as variation of BS intensities has no impact on V2V communication. 4) it is expected that in current heterogeneous networks where the intensity of BS is mostly high and sufficient, the number of cellular BSs are available to provide highly reliable coverage to vehicular networks. From this analysis, we can conclude that C-V2X communication with both direct and network modes is going to provide better reliability and coverage performance than V2V communication alone.

		%
		
		\begin{figure}
			\begin{center}
				\includegraphics[scale=0.42]{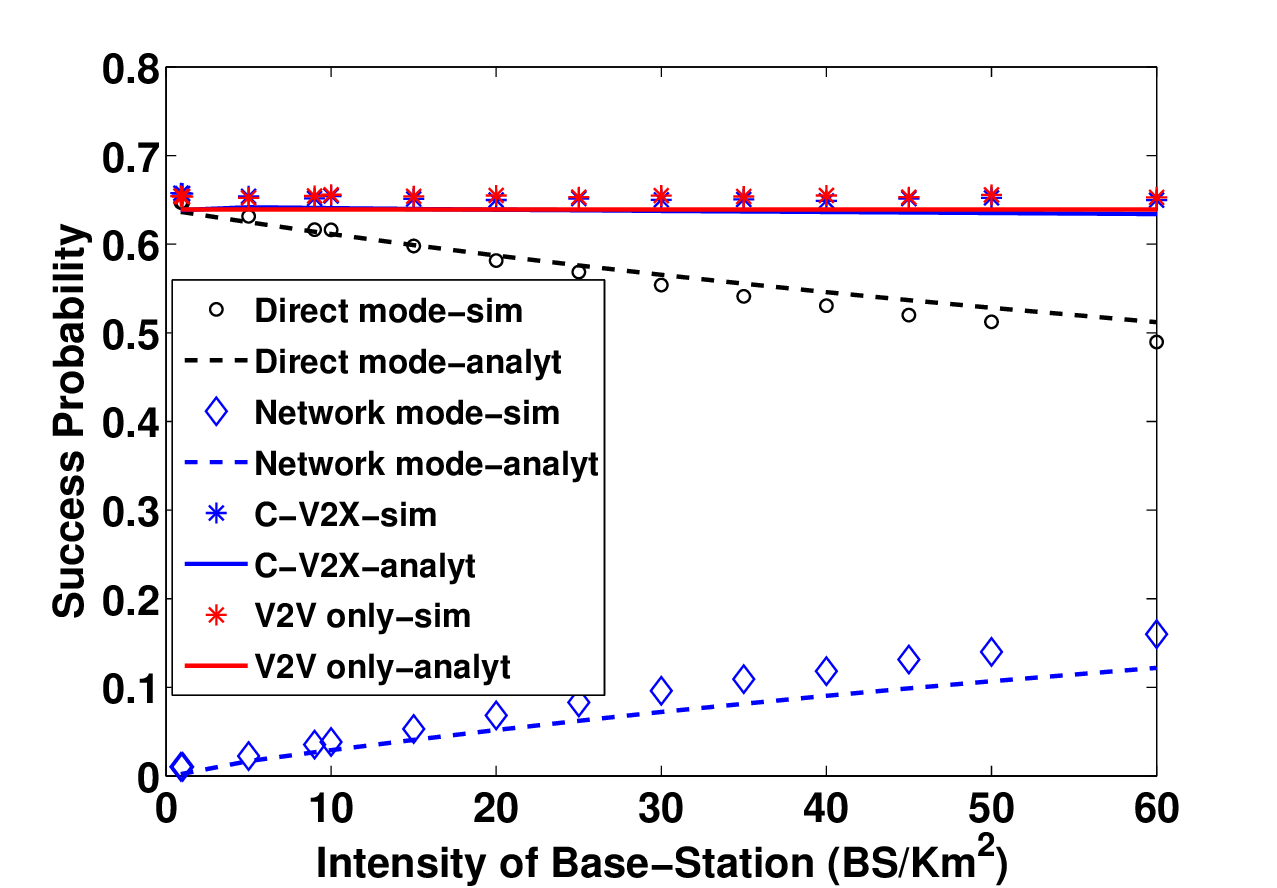}
				\caption{{The success probability versus the BSs intensity.}}
				\label{BSVarOverall}
			\end{center}
		\end{figure}

		
		\subsection{Impact of the association bias}
		In this subsection, we examine the effect of association bias, $B$ at success probability of C-V2X network and V2V communication. In Fig. \ref{Bias}, we set {$\mu_v$ = 5 nodes/km, $\lambda_R$ = 5 km/km$^2$, $\lambda_b = 20$ $BSs/km^2$ and $z=0$ dB.} 
		
		Fig. \ref{Bias} plots the success probability at the arbitrary receiver versus the association bias. The following insights can be observed: 1) the success probability of C-V2X communication is slightly higher than V2V communication alone as both direct and network modes supplement each other for success probability of C-V2X communication and there are no coverage gaps for the network. 2) we see that at lower bias values, there is an equal opportunity for both modes being selected. At $B=1$, there is equal probability of association with direct and network modes. 3) {the success probability of V2V communication without mode selection scheme remains constant as association bias has no effect on this type of communication.} 4) we can summarize that traffic loading and interference to cellular networks can be controlled through association bias. 
		
		%
		\begin{figure}
			\begin{center}
				\includegraphics[scale=0.42]{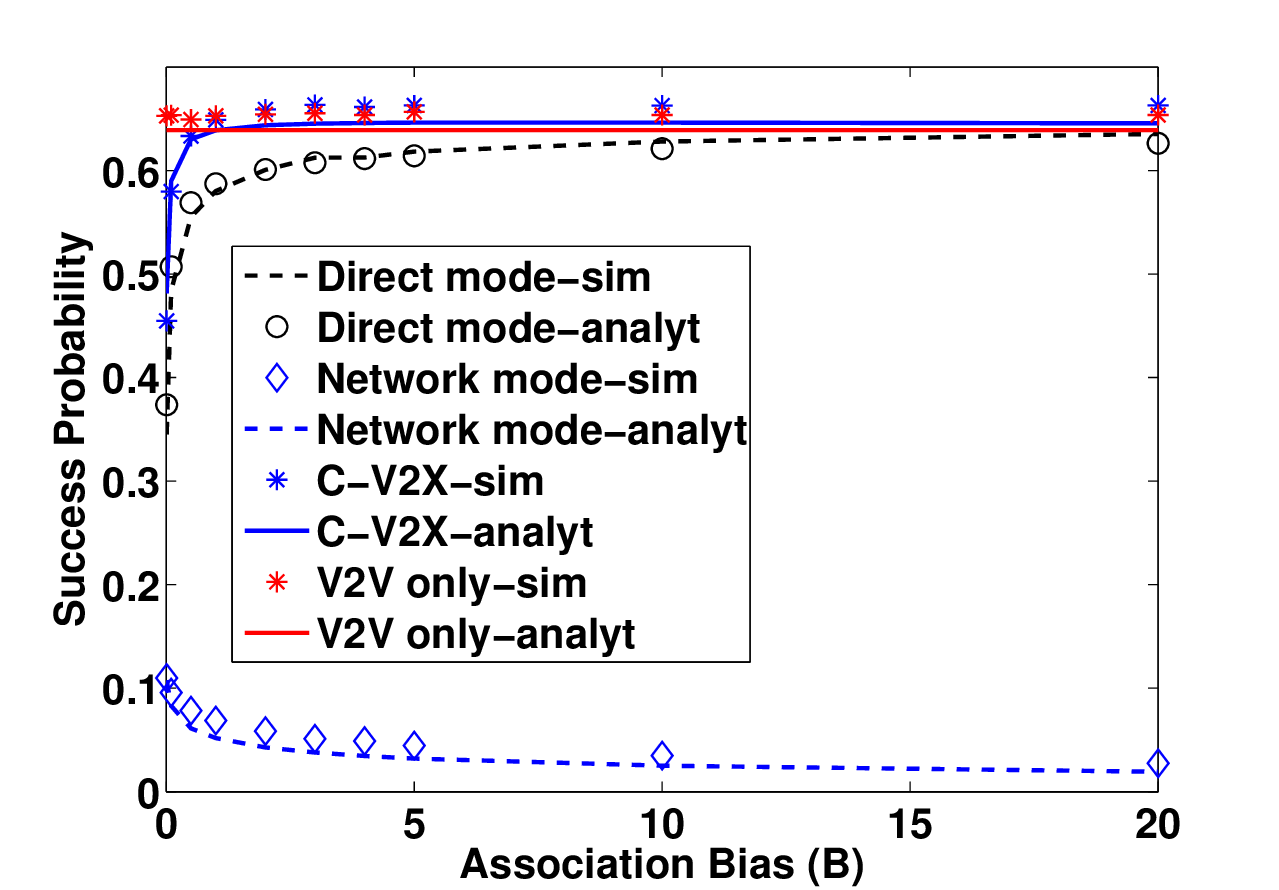}
				\caption{{The success probability versus the association bias.}}
				\label{Bias}
			\end{center}
		\end{figure}

		\section{Conclusion}
		{In this paper, we presented a comprehensive and tractable analytical framework for the reliability performance of cellular vehicle-to-everything (C-V2X)  communication in which vehicular communication can be established through cellular network or directly between vehicles on shared V2V and cellular uplinks channels. A flexible mode selection scheme has been proposed for the vehicular transmitter to decide between the direct and network modes, with a bias factor controlling the amount of vehicular interference and traffic on the cellular network. By modeling the vehicles on roads as doubly stochastic Cox process, and the BSs as 2D PPP, we derived the expressions for the success probabilities of the direct mode, the network mode (both uplink and downlink), as well as the C-V2X link having both modes, which are validated by the simulations. By comparing the proposed mode based C-V2X communication with the solely V2V communication without mode selection, we have shown the reliability enhancement brought by the shared communication via cellular networks along with spectral efficiency by sharing V2V link and cellular uplink frequencies. Future works can be extended to 1) interference mitigation techniques for cellular V2X network 2) the derivation of V2V communication analytical model for Ricean fading channel and results comparison with Rayleigh fading channels 3) modeling of vehicles on PLP as hard-core process on the line instead of PPP 4) implementation of the proposed model for vehicular communication and validation of reliability results of C-V2X communication in comparison with V2V communication by the industry 5) correctness verification of using PLP for stochastic modeling the roads by using the Google maps and real roads.}

		\appendices

		\section{Proof of Lemma \ref{Lemma2}}
		\label{AppendixB}
		The conditional association probability of the V2V communication for shared link is given by
		\begin{equation}
		{P_{v2v}^{A}}(v2v|{R_v}) = P\left[ {B \times r_v^{ - {\alpha _v}} \ge r_b^{ - {\alpha _b}}} \right],\label{ProbV2Vshared}
		\end{equation}
		where $B$ is the association bias, $r_v$ is the V2V link distance and $r_b$ is the V2B link distance. By simplifying the above equation, we get	
		\begin{equation}
		{P_{v2v}^{A}}(v2v|{R_v}) = P\left[ {{r_b} \ge \frac{{r_v^{\frac{{{\alpha _v}}}{{{\alpha _b}}}}}}{{{B^{\frac{1}{{{\alpha _b}}}}}}}} \right].\label{ProbV2Vshared2}
		\end{equation}
		
		It means that the vehicular transmitter which is located at distance $r_v$ is connected to vehicular receiver. Therefore, the disk B(0, ${\frac{{r_v^{\frac{{{\alpha _v}}}{{{\alpha _b}}}}}}{{{B^{\frac{1}{{{\alpha _b}}}}}}}}$) does not contain any BS. By inserting CCDF of V2B link in (\ref{ProbV2Vshared2}), we get
		\begin{equation} \label{V2Vpdfconditioned}
		{P_{v2v}^{A}}(v2v|{R_v}) = \exp \left( { - \pi {\lambda _B}{{\left( {\frac{{r_v^{\frac{{{\alpha _v}}}{{{\alpha _b}}}}}}{{{B^{\frac{1}{{{\alpha _b}}}}}}}} \right)}^2}} \right).
		\end{equation}

		Now, removing condition on $R_v$, we get
		\begin{equation} \label{V2Vpdffinal}
		{P_{v2v}^{A}} = \int\limits_0^\infty  {\exp \left( { - \pi {\lambda _B}{{\left( {\frac{{r_v^{\frac{{{\alpha _v}}}{{{\alpha _b}}}}}}{{{B^{\frac{1}{{{\alpha _b}}}}}}}} \right)}^2}} \right) \times {f_{{R_v}}}({r_v})d{r_v}}.
		\end{equation}
		

		
		The PDF  of $R_v$ ${f_{{R_v}}(r_v)}$ is given in Eq. (\ref{pdfv2vfinalclosedform}). The solution for the V2V association probability is proved in Eq. (\ref{V2VProbincomplete}).

		\section{Proof of Lemma \ref{Lemma3}}
		\label{AppendixC}
		The conditional association probability of the V2B shared link is given as
		\begin{equation}\label{V2Basseq11}
		{P_{v2b}^{A}} (v2b|{R_b}) = P\left[ {B \times r_v^{ - {\alpha _v}} < r_b^{ - {\alpha _b}}} \right].
		\end{equation}
		
		Simplifying the above equation, we get
		\begin{equation}\label{V2Basseq111}
		{P_{v2b}^{A}}(v2b|{R_b}) = P\left[ {{r_v} > {B^{\frac{1}{{{\alpha _v}}}}}r_b^{\frac{{{\alpha _b}}}{{{\alpha _v}}}}} \right].
		\end{equation}
		
		It means that the disk b(0, ${{B^{\frac{1}{{{\alpha _v}}}}}r_b^{\frac{{{\alpha _b}}}{{{\alpha _v}}}}}$) does not contain any vehicle, and thus	
		\begin{equation} \label{V2Basseq1}
		{P_{v2b}^{A}}(v2b|{R_b}) = P\left[ {{\rm{\text{No vehicle in circle of radius},}}\,{B^{\frac{1}{{{\alpha _v}}}}}r_b^{\frac{{{\alpha _b}}}{{{\alpha _v}}}}} \right].
		\end{equation}
		
		%
		%

		By using the CCDF of $R_v$ given in Eq. (\ref{cdfv2v}), we get
		\begin{align} \label{V2Basseq}
		{P_{v2b}^{A}}(v2b|{R_b}) &= {\exp \Bigg[ { - 2\pi {\lambda _R}} \Bigg.}\nonumber
		\\&\Bigg. {\int\limits_{{y_n} = 0}^{{B^{\frac{1}{{{\alpha _v}}}}}r_b^{\frac{{{\alpha _b}}}{{{\alpha _v}}}}} {1 - {e^{ - 2{\mu _v}\sqrt {{{\left( {{B^{\frac{1}{{{\alpha _v}}}}}r_b^{\frac{{{\alpha _b}}}{{{\alpha _v}}}}} \right)}^2} - {y_n}^2} }}} d{y_n}} \Bigg]\nonumber
		\\&\times {\exp({ - 2{\mu _v}{B^{\frac{1}{{{\alpha _v}}}}}r_b^{\frac{{{\alpha _b}}}{{{\alpha _v}}}}}}).
		\end{align}
		
		
		Now, by removing condition on $R_b$ and by inserting CDF of $R_b$ in (\ref{V2Basseq}), we get Eq. (\ref{V2BAssProb}) for the unconditional V2B link association probability over shared channels.
		
		\section{Proof of Corollary \ref{cor1}}
		\label{AppendixD}
		
		The length of road $R_{out}$ lying inside the circular region $\mathcal{B}(0,R)$ is $2\sqrt {{R}^2 - {y}^2}$ and distance between two vehicles can be denoted as $t$. By using the properties of a PPP, the probability of there being $m$ points in this line segment can be calculated from 1D PPP. The conditional Laplace transform from these vehicles lying on this road segment to typical receiver can be calculated as follows 
		\begin{equation} \label{Laplacesingleroad1}
		{{\cal L}_{{I_{{R_{out}}}}}}\left( {s|r} \right) = {E_{{D_x}}}\left[ {\prod\limits_{x \in {R_{out}}} {\left[ {\frac{1}{{1 + s{P_v}D_{^x}^{ - \alpha }}}} \right]} } \right].
		\end{equation}
		
		Now conditioning over number of vehicles lying on the road and then deconditioning on interference caused by each node, we get
		\begin{align} \label{Laplacesingleroad2}
		{{\cal L}_{{I_{{R_{out}}}}}}\left( {s|r} \right) &= \sum\limits_{m \ge 0} {P\left[ {{N_v} = m} \right]}\nonumber
		\\&\times {\left[ {\int\limits_{  - \sqrt {{R^2} - {y^2}} }^{\sqrt {{R^2} - {y^2}} } {\frac{{f(t)dt}}{{1 + s{P_v}{{\left( {{y^2} + {t^2}} \right)}^{\frac{{ - \alpha }}{2}}}}}} } \right]^m}.
		\end{align}
		
		The number of points on the line segment of length ${2\sqrt {{R^2} - {y^2}} }$ is a Poisson random variable with mean ${2{\mu _v}\sqrt {{R^2} - {y^2}} }$ and $t$ is uniformly distributed between $(-\sqrt {{R}^2 - {y}^2}, \sqrt {{R}^2 - {y}^2})$ and has a PDF $f(t) = \frac{1}{{2\sqrt {{R^2} - {y^2}} }}$. By inserting pdf of $f(t)$ in the above equation, we get
		\begin{align} \label{Laplacesingleroad3}
		{{\cal L}_{{I_{{R_{out}}}}}}\left( {s|r} \right) &= \sum\limits_{m \ge 0} {\frac{{{e^{ - 2{\mu _v}\sqrt {{R^2} - {y^2}} }}{{\left( {2{\mu _v}\sqrt {{R^2} - {y^2}} } \right)}^m}}}{{m!{{\left( {2\sqrt {{R^2} - {y^2}} } \right)}^m}}}}\nonumber
		\\&\times {\left[ {\int\limits_{ - \sqrt {{R^2} - {y^2}} }^{\sqrt {{R^2} - {y^2}} } {\left( {\frac{{dt}}{{1 + s{P_v}{{\left( {{y^2} + {t^2}} \right)}^{\frac{{ - \alpha }}{2}}}}}} \right)} } \right]^m}.
		\end{align}
		
		By simplifying above equation by using the property of integral of even function and by using the Taylor series expansion, ${e^x} = \sum\limits_{n = 0}^\infty{\frac{{{x^n}}}{{n!}}}$, we get	
		\begin{align} \label{Laplacesingleroad4}
		{{\cal L}_{{I_{{R_{out}}}}}}\left( {s|r} \right) &= \exp \left( { - 2{\mu _v}\sqrt {{R^2} - {y^2}} } \right)\nonumber
		\\&\times \exp \left( {2{\mu _v}\int\limits_{ 0}^{\sqrt {{R^2} - {y^2}} } {\left( {\left[ {\frac{{dt}}{{1 + s{P_v}{{\left( {{y^2} + {t^2}} \right)}^{\frac{{ - \alpha }}{2}}}}}} \right]} \right)} } \right).
		\end{align} 
		
		The final expression for Laplace transform of interference from a road located outside the inner circular region $b(0,r)$ is given in equation (\ref{LaplaceSingleRoad2}).

		\section{Proof of Corollary \ref{cor8}}
		\label{AppendixE}
		We condition on number of roads, $j$ crossing the region $\mathcal{B}(0,r)$ to calculate the interference from vehicles which lie on roads which interest the region $\mathcal{B}(0,r)$. However, on these lines, the interfering vehicles are located outside the region $\mathcal{B}(0,r)$. This is a Poisson random variable with mean ${2{\mu _v}r{\lambda _R}}$. Similarly, we condition on number of roads, $k$ which lie between circular regions $\mathcal{B}(0,r)$ and $\mathcal{B}(0,R)$. This is also a Poisson random variable with mean ${2{\mu _v}(R-r){\lambda _R}}$. Therefore, the Laplace Transform of interference of the total interference originating from Poisson Line Process can be written as	
		\begin{align} \label{Laplaceallroad1}
		{\mathcal{L}_{{I_v}}}\left( {s|r} \right) &= \left[ {\sum\limits_{j \geqslant 0} {\frac{{\exp \left( { - 2{\mu _v}d{\lambda _R}} \right) \times {{\left( {2{\mu _v}r{\lambda _R}} \right)}^j}}}
				{{j!}}} } \right.\nonumber
		\\&\left. { \times \left[ {{{\left( {\int\limits_{ - r}^r {{\mathcal{L}_{{I_{R1}}}}\left( s \right) \times {f_Y}(y)dy} } \right)}^j}} \right]} \right]\nonumber
		\\&\times \left[ {\sum\limits_{k \geqslant 0} {\frac{{{e^{ - 2{\mu _v}\left( {R - r} \right){\lambda _R}}} \times {{\left( {2{\mu _v}\left( {R - r} \right){\lambda _R}} \right)}^k}}}
				{{k!}}} } \right.\nonumber
		\\&\left. { \times \left[ {{{\left( {\int\limits_{y=r}^R {{\mathcal{L}_{{I_{R2}}}}\left( s \right) \times {f_Y}(y)dy} } \right)}^k}} \right]} \right].
		\end{align}   
		
		%
		
		By writing the above equation in the form of Taylor series,
		we have	
		\begin{align} \label{Laplaceallroad3}
		{{\cal L}_{{I_v}}}\left( {s|r} \right) &= \left[ {{e^{ - 2{\mu _v}r{\lambda _R}}}\sum\limits_{j \ge 0} {\frac{{{{\left( {{\mu _v}{\lambda _R}\int\limits_{y =  - r}^r {{{\cal L}_{{I_{R1}}}}\left( s \right)dy} } \right)}^j}}}{{j!}}} } \right] \nonumber\\ 
		&\times \left[ {{e^{ - 2{\mu _v}\left( {R - r} \right){\lambda _R}}}\sum\limits_{k \ge 0} {\frac{{{{\left( {{\mu _v}{\lambda _R}\int\limits_{y = r}^R {{{\cal L}_{{I_{R2}}}}\left( s \right)dy} } \right)}^k}}}{{k!}}} } \right].
		\end{align} 
		
		By simplifying the above equation, we get the final result for Laplace Transform of interference of vehicles located on all roads excluding the road passing through the origin and it is proved in Eq. (\ref{LaplaceAllRoads}).

		%
		%

		\ifCLASSOPTIONcaptionsoff
		\newpage
		\fi

		\bibliographystyle{IEEEtran}
		\bibliography{Bibliography}
		
		%
		%

		

	\end{document}